\documentclass[12pt,twoside]{report}

\usepackage{amsmath}
\newlength{\boxwidth}
\newcommand{\qed}{\hfill $\Box$}
\newcommand{\ignore}[1]{}

\RequirePackage{ifthen}

\newcommand{\BdefTemplate}[2][*]{
 \newboolean{#2}
 \newboolean{#2_MyClass_Switch}
 \setboolean{#2_MyClass_Switch}{true}
 \ifthenelse{\equal{unset}{#1}}{\setboolean{#2}{false}}{
  \ifthenelse{\equal{set}{#1}}{\setboolean{#2}{true}}{
   \ifthenelse{\equal{}{#1}}
    {\setboolean{#2_MyClass_Switch}{false}}
    {\ClassWarning{MyClass}{Boolean variable <#2> uninitialized}}}}}

\newcommand{\Bdef}{\BdefTemplate}

\newcommand{\Bset}[1]{\setboolean{#1}{true}}
\newcommand{\Bunset}[1]{\setboolean{#1}{false}}

\Bdef[unset]{DocumentStarted}
\newcommand{\IfDocumentStarted}{\ifthenelse{\boolean{DocumentStarted}}}
\AtBeginDocument{\Bset{DocumentStarted}}

\Bdef[unset]{Rough}
\newcommand{\IfRough}{\ifthenelse{\boolean{Rough}}}

\newcommand{\Bif}[3]{\ifthenelse{\boolean{#1_MyClass_Switch}}
{\ifthenelse{\boolean{#1}}{#2}{#3}}
{\ifthenelse{\not \boolean{Rough}}
  {\ClassError{MyClass}
   {Boolean variable <#1> splitting in non-rough mode}
   {In a non-rough mode, you may not use variable splitting (\Bdef[]...)}}{}
\IfDocumentStarted                   
{\begin{description}
  \item[\framebox{#1 == True :}] #2
  \item[\framebox{#1 == False :}] #3
  \item[\framebox{End of #1}]
 \end{description}}
{#2 #3}
}}

\newcommand{\comment}[1] {}
\newcommand{\eps}{\epsilon}
\newcommand{\cJ}{{\cal J}}
\global\long\def\a{\alpha}
\newcommand{\cJz}{{\cal J}_0}
\usepackage{algorithm}
\usepackage{epsfig}
\usepackage{epstopdf}

\Bdef[unset]{non_Class}

\Bdef[unset]{extd_abstr}

\newcommand{\IfAbstr}{\Bif{extd_abstr}}

\Bdef[unset]{Abib_new_page}

\newcommand{\IfAbNP}{\Bif{Abib_new_page}}

\Bdef[unset]{is_llncs}
\newcommand{\Ifllncs}{\Bif{is_llncs}}

\Bdef[unset]{is_jmlr}
\newcommand{\Ifjmlr}{\Bif{is_jmlr}}

\Bdef[unset]{enum_refs}

\Bdef[unset]{is_slides}

\Bdef[unset]{is_report}

\Bdef[unset]{is_letter}

\Bdef[unset]{use_graphicx}

\Bdef[unset]{in_math_mode}
\newcommand{\IfMathMode}{\Bif{in_math_mode}}
\newcommand{\MathModeOn}{\Bset{in_math_mode}}
\newcommand{\MathModeOff}{\Bunset{in_math_mode}}

\Bdef[unset]{import_mode}

\Bset{Rough}
\Bset{import_mode}

\Ifjmlr{
 \RequirePackage{jmlr2e}  
}{
 \RequirePackage{latexsym}
 \RequirePackage{amssymb}
} 

\newenvironment{MyBlank}[1][]{#1}{}

\usepackage{amsthm}
\usepackage{makeidx}

\Ifllncs{}{
 \Ifjmlr{}{
  \newtheorem{theorem}{Theorem}[chapter]
  \newtheorem{lemma}[theorem]{Lemma}
  \newtheorem{corollary}[theorem]{Corollary}
  
  \newtheorem{proposition}[theorem]{Proposition}
  \newtheorem{property}[theorem]{Property}
  \newtheorem{definition}{Definition}[chapter]

 } 
 \newtheorem{claim}[theorem]{Claim}
 
} 
\newtheorem{fact}[theorem]{Fact}

\newcommand{\nospell}[1]{#1}

\newcommand{\newident}[3][*]{\ifthenelse{\equal{*}{#1}}
 {\newcommand{#2}[1][*]
  {\ifthenelse{\equal{*}{##1}}
   {\nospell{\mbox{\Ensuremath{#3}}}}
   {#3}}}
 {\newcommand{#2}{#3}}}

\newcommand{\newmat}[3][*]{\ifthenelse{\equal{*}{#1}}
 {\newcommand{#2}[1][*]
  {\ifthenelse{
   \( \equal{*}{##1} \and \not \boolean{in_math_mode} \)
   \or \( \not \equal{*}{##1} \and \boolean{in_math_mode} \)
  }
   {\nospell{\mbox{\Ensuremath{#3}}}}
   {#3}}}
 {\newcommand{#2}{#3}}}

\newcommand{\bibentry}[7]{
\Ifjmlr
{\bibitem[#2(#7)]{#1}
 {\textup #3}. \newblock \textrm{#4.} \newblock {\em #5, #6}, #7.}
{\ifthenelse{\boolean{is_llncs} \or \boolean{enum_refs}}
 {\bibitem{#1} {\textup #3}.}
 {\bibitem[\nospell{#1}]{#1} {\textup #3}.}
\newblock \textrm{#4.} \newblock {\em #5, #6}, #7.}}

\newcommand{\bib}[1][]{
\Ifjmlr{
 \begin{thebibliography}{2}
   \in put{bib}
   #1
 \end{thebibliography}
}{
 \begin{thebibliography}{ABCDE00}
   \frenchspacing
   \in put{bib}
   #1
 \end{thebibliography}}}

\Ifjmlr
 {}
 {
  \newcommand{\citep}[1]{\cite{#1}}}

\newcommand{\Abib}[1][]{
 \IfAbstr{}{
  \Ifllncs{}{
   \IfAbNP{
    \newpage\clearpage}
   {}}}
 \bib[#1]}

\newcommand{\nop}[1]{}     

\Bif{is_slides}{

}{}

\newcommand{\Ensuremath}[1]{\IfMathMode
 {\ensuremath{#1}}
 {\MathModeOn\ensuremath{#1}\MathModeOff}}

\newmat{\D}{{D}}

\reversemarginpar

\newcommand{\cM}{{\cal M}}

\newcommand{\itab}[1]{\hspace{0em}\rlap{#1}}
\newcommand{\tab}[1]{\hspace{.2\textwidth}\rlap{#1}}

\newcommand{\singlespacing}%
{\small\normalsize}

\newident{\Ex}{Ex}
\newident{\EQ}{EQ}
\newident{\R}{R}
\newident{\Q}{Q}
\newident{\W}{W}

\begin{document}

\begin{titlepage}
\begin{center}

\Huge \textbf{Optimization and Reoptimization in Scheduling Problems}
\vfill

\Huge
\textbf{Yael Mordechai}\\
\vfill

\end{center}
\end{titlepage}

\begin{titlepage}
~
\end{titlepage}

\begin{titlepage}
\begin{center}

\Huge \textbf{Optimization and Reoptimization in Scheduling Problems}
\vfill

\large
Research Thesis\\
\vfill

Submitted in Partial Fulfillment of the Requirements
for the Degree of Master of Science in Computer Science\\
\vfill

\Huge
\textbf{Yael Mordechai}\\
\vfill

\large
Submitted to the Senate of the\\
Technion - Israel Institute of Technology\\
\vfill

{\center Tamuz 5775  \hfill   Haifa  \hfill   June 2015}
\thispagestyle{empty}
\end{center}
\end{titlepage}
\thispagestyle{empty} 
\begin{titlepage}
\begin{center}
This research was carried out under the supervision of Prof. Hadas Shachnai,
in the Computer Science Department.
\end{center}
\vfill
\thispagestyle{empty} 
\begin{center}
The generous financial help of the Technion is gratefully
acknowledged.
\end{center}
\end{titlepage}
\thispagestyle{empty} 
\thispagestyle{empty}
\tableofcontents \thispagestyle{empty} 
\addtocontents{toc}{\protect\thispagestyle{empty}}
\thispagestyle{empty}
\newpage
\thispagestyle{empty}
\listoffigures
\thispagestyle{empty}
\listoftables
\thispagestyle{empty}
\newpage
\thispagestyle{empty}
\chapter*{Abstract}
\addcontentsline{toc}{chapter}{Abstract}
\setcounter{page}{1}
Parallel machine scheduling has been extensively studied in the past decades, with
applications ranging from production planning to job processing in large computing 
clusters. In this work we study some of these fundamental optimization problems, as well as their
parameterized and reoptimization variants.

We first present improved bounds for job scheduling on unrelated parallel machines, 
with the objective of minimizing the latest completion time (or, makespan) of the schedule. 
We consider the subclass of {\em fully-feasible} instances, in which the processing time of each 
job, on any machine, does not exceed the minimum makespan. The problem is known to be hard 
to approximate within factor 4/3 already in this subclass. Although fully-feasible instances 
are hard to identify, we give a polynomial time algorithm that yields for such instances a
schedule whose makespan is better than twice the optimal, the best known ratio for general 
instances. Moreover, we show that our result is robust under small violations of feasibility 
constraints.

We further study the power of parameterization. In a parameterized optimization problem, each input comes with a fixed parameter. Some problems can be solved by algorithms (or approximation algorithms) that are exponential only in the size of the parameter, while polynomial in the input size. The problem is then called {\em fixed parameter tractable (FPT)}, since it can be solved efficiently (by an FPT algorithm or approximation algorithm) for constant parameter values.
We show that makespan minimization on unrelated machines admits a {\em parameterized approximation scheme}, where the parameter 
used is the number of processing times that are large relative to the latest completion 
time of the schedule. We also present an FPT algorithm for the graph-balancing problem, which corresponds to the instances of the {\em restricted assignment} problem where each job can be processed on at most 2 machines.

Finally, motivated by practical scenarios, we initiate the study of 
{\em reoptimization} in job scheduling on identical and uniform machines,
with the objective of minimizing the makespan. We develop {\em
reapproximation} algorithms that yield in both models the best possible 
approximation ratio of $(1+\eps)$, for any $\eps >0$, with respect to the
minimum makespan.

\newpage

\newpage
\comment{
\chapter*{List of Symbols and Abbreviations}

\begin{itemize}
\item \itab{$n$}             \tab{}\tab{number of jobs}
\item \itab{$m$}             \tab{}\tab{number of machines}
\item \itab{$G$}             \tab{}\tab{a graph}
\item \itab{$V$}             \tab{}\tab{the set of the graph vertices}
\item \itab{$E$}             \tab{}\tab{the set of the graph edges}
\item \itab{$[n]$}           \tab{}\tab{$\{1,2,\cdots,n\}$}
\end{itemize}
}

\newpage

\chapter{Introduction}

\section{Scheduling on Parallel Machines }
Consider the following fundamental problem in scheduling theory. We are given a set $\cal J$ of $n$ independent jobs that must be scheduled without preemption on a collection $\cal M$ of $m$ parallel machines. If job $j$ is scheduled on machine $i$, the processing time required is $p_{ij}$,
which is a positive integer, for every $i \in \cal M$ and $j\in \cal J$. The total time used by machine $i\in \cal M$, or the {\em load} on machine $i$, is the sum of the processing times for the jobs assigned to $i$, and the {\em makespan} of an assignment is the maximum load over all the machines. The objective is then to find a schedule, which assigns each job to exactly one machine, such that the makespan is minimized.

\subsection{Scheduling Models}
The wide literature on scheduling often distinguishes between  the following scheduling models.
\paragraph*{Identical Machines.} Job processing times are identical across the machines, i.e., $p_{ij}=p_j$ for all $j\in \cal J$ and $i\in \cal M$.
\paragraph*{Uniform Machines.} Each machine $i$ has a speed $s_i$. The length of job $j$ on machine $i$ is some uniform processing time $p_{j}$ scaled by the speed $s_i$, i.e., $p_{ij}=\frac{p_j}{s_i}$ for all $j\in J$ and $i\in \cal M$.
\paragraph*{Unrelated Machines:} Each job $j$ may have an arbitrary processing time $p_{ij} \geq 0$ on
machine $i$, for $j\in \cal J$ and $i\in \cal M$.
\\
\\
While makespan minimization is known to be NP-hard in all models (even for $m=2$) \cite{LK79},
the first two models are considered somewhat easier, since the problem can be approximated efficiently in both up to some $\epsilon$ factor, for any $\epsilon>0$ (see Section \ref{Chapter1 related work}). In contrast, in the unrelated machines model, the problem becomes hard to approximate within a factor better than $\frac{3}{2}$. Moreover, since 1990, when the state of the art $2$-approximation algorithm was presented by Lenstra, Shmoys and Tardos \cite{LST90}, there was no significant improvement on either the upper or lower bound, although the problem was consistently investigated. This led researchers to consider special cases and improving the bound of $2$, either by a constant factor, or by some function of the input parameters (see review in Section \ref{Chapter1 related work}).

In this work, we consider the subclass of {\em fully-feasible} instances. We say that an instance is fully-feasible if job processing times, $p_{ij}$, do not exceed the length of the optimal schedule for the instance, for every job $j\in \cal J$ and machines $i\in \cal M$. Observe that an optimal schedule never assigns a job to machine on which its length is greater than the makespan of an optimal schedule; thus, from an optimal scheduler's viewpoint, if $p_{ij}$ exceeds the optimal makespan then $p_{ij}$ is considered to be $\infty$.

We also consider instances that are almost fully-feasible, that is, for any job $j$, the number of machines on which job $j$ is not feasible (i.e., has processing time larger than the length of the optimal makespan) is relatively small. 

When considering real-life applications, the general model of unrelated machines, which makes no assumptions on job processing times, seems too broad. Indeed, such applications usually deal with fully-feasible (or almost fully-feasible) workloads, as they commonly handle relatively large sets of jobs. 

Let $T_{opt}$ and $L_{opt}$ denote the optimal makespan and the minimal average machine load over optimal assignments, respectively. For heterogeneous workloads of a huge number of jobs, in which the makespan is counted in months or even years, the processing time of a given job is negligible compared to the makespan; for such workloads, we have $p_{ij} \ll T_{opt}$. In this case, an algorithm of \cite{ST93}, yields a schedule of makespan at most $T_{opt}+p_{max}$, where $p_{max}=max_{ij}p_{ij}$, is more suitable. However, for smaller sets of jobs, $p_{ij}$ can be large relative to $T_{opt}$, such that $L_{opt}<p_{max}$, and for these instances our algorithm is the state of the art.
Relevant applications for such workloads are e.g., job packing in warehouse-scale \cite{VKW14}, large-scale clustering \cite{VP+15} and applications in parallel design patterns such as Fork-Join and MapReduce (see, e.g., \cite{DMN12,LL+12}).

\subsection{Related Work}
\label{Chapter1 related work}

\comment{\paragraph*{Identical and Uniform Machines} The problem of makespan minimization on identical or uniform machines is known to be NP-hard in the strong sense already in the identical machines model (as it contains bin packing and 3-Partition as special cases) \cite{GJ79}. Therefore, we cannot hope for an FPTAS. For identical machines, Hochbaum \cite{H96} and Alon et al. \cite{AA+98} gave an EPTAS with running time $f(\frac{1}{\epsilon})+O(n)$, where $f$ is doubly exponential in $\frac{1}{\epsilon}$, and Jansen \cite{J10} gave an EPTAS with running time $2^{O(1/ \epsilon^2 log(1/\epsilon)^3)} + poly(|{\cal I}|)$.
}
\paragraph*{Identical and Uniform Machines} The problem of makespan minimization on identical or uniform machines is known to be NP-hard \cite{GJ79}. A {\em polynomial-time approximation scheme (PTAS)} is a family of algorithms $\{A_\epsilon:\epsilon >0\}$, where $A_\epsilon$ is a $(1+\epsilon)$-approximation algorithm that runs in time polynomial in the input size but is allowed to be exponential in $\frac{1}{\epsilon}$. An {\em efficient polynomial-time approximation scheme (EPTAS)} is a PTAS with running time $f(\frac{1}{\epsilon})poly(|{\cal I}|)$, where $|{\cal I}|$ is the input size, (for some function $f$), while a {\em fully polynomial-time approximation scheme (FPTAS)} runs in time $poly(\frac{1}{\epsilon}, |{\cal I}|)$. Since the scheduling problem is NP-hard in the strong sense already on identical machines (as it contains bin packing and 3-Partition as special cases) \cite{GJ79}, we cannot hope for an FPTAS. For identical machines, Hochbaum \cite{H96} and Alon et al. \cite{AA+98} gave an EPTAS with running
time $f(\frac{1}{\epsilon})+O(n)$, where $f$ is doubly exponential in $\frac{1}{\epsilon}$, and for uniform machines, Jansen \cite{J10} gave an EPTAS with running time $2^{O(1/ \epsilon^2 log(1/\epsilon)^3)} + poly(|{\cal I}|)$.

\paragraph{Unrelated Machines} A classic result in scheduling theory is the Lenstra-Shmoys-Tardos $2$-approximation algorithm for makespan minimization \cite{LST90}. They also proved that the problem is NP-hard to approximate within a factor better than $\frac {3}{2}$.
Gairing et al. \cite{GMW07} presented a more efficient, combinatorial $2$-approximation algorithm based on flow techniques. 
Shchepin and Vakhania \cite{SV05} showed that the rounding technique used in \cite{LST90} can be modified to derive an improved ratio of $2-\frac{1}{m}$. Shmoys and Tardos \cite{ST93} showed an approximation algorithm that yields a schedule of makespan at most the length of an optimal schedule plus the largest processing time of any job in the instance. 

Although makespan minimization on unrelated machines is a major open problem in scheduling theory, and is extensively studied, there was no significant progress on either the upper or lower bound for over two decades, since the publication of [LST90]. This led researchers to consider interesting special cases and improving the upper bound for them. A well known special case is the {\em restricted assignment problem}, where jobs have processing times $p_{ij}\in \{p_j,\infty\}$. Svensson \cite{S12} gave a polynomial-time algorithm that approximates the optimal makespan of
the restricted assignment problem within a factor of $\frac {33}{17}+\epsilon \approx 1.94+\epsilon $ for $\epsilon >0$, and also presented a local search algorithm that will eventually find a schedule of the mentioned approximation guarantee, but is not known to converge in polynomial-time.
Gairing et al. \cite{GL+04} presented a combinatorial $(2-\frac{1}{p_{max}})$-approximation algorithm for the restricted assignment problem. 

Ebenlendr et al. considered in \cite{EKS08} the {\em graph balancing} problem, a special case of the restricted assignment problem where each job $j$ has a finite processing time, $p_{ij} < \infty$, on at most two machines. The paper gives an elaborate $1.75$-approximation algorithm for the
problem. The authors also show that the problem is hard to approximate within a factor less than $\frac{3}{2}$ even on bounded degree graphs, i.e., when the maximum degree is some constant.  
In the {\em unrelated graph balancing} problem, introduced by Versache
and Weiss \cite{VW14}, each job can be assigned to at most two machines, but processing times are not restricted. They showed that this subclass of instances constitutes the core difficulty for the {\em linear programming} formulation of makespan minimization on unrelated machines, often used as a first step in obtaining approximate solutions. Specifically, they showed that the strongest known LP-formulation, namely, the configuration-LP, has an integrality gap of $2$. 

Vakhania et al. \cite{VHW14} considered makespan minimization on unrelated machines for the subclass of instances where job lengths can take only two values, $p$ and $q$, which are fixed
positive integers, such that $p<q$. They presented a polynomial-time algorithm that uses linear programming with absolute approximation factor of $q$ (i.e., all schedules have makespan at most $OPT+q$). Page \cite{P14} considered restricted assignment instances with processing times in a fixed interval, $[p,q]$, and gave a $\frac {q}{p}$-approximation algorithm, and a $\frac{3}{2}$-approximation algorithm for the case where $p_{ij} \in \{1,2,3\}$. Chakrabarty et al. \cite{CK+15} considered instances with two types of jobs: {\em long} and {\em short}, namely, $p_{ij}\in \{1,\epsilon\}$ for some $\epsilon>0$. They obtained a $(2-\delta)$-approximation algorithm for such instances. 

Shmoys and Tardos \cite{ST93} considered the {\em generalized assignment problem (GAP)}, where each job $j$ incurs a cost of $c_{ij}>0$ when assigned on machine $i$, and the objective is to minimize the makespan and the total cost.
The paper \cite{ST93} presents a polynomial-time algorithm that finds a schedule of makespan at most twice the optimum with optimal cost.

A summary of the known results for unrelated machines is given in Table \ref{table_unrelated}.

\begin{table}[H]
\centering 
\begin{tabular}{|c|c|c|} 
\hline
Result & Authors & Restrictions on the Unrelated Model \\ [0.5ex] 
\hline\hline
$2$-approximation &\cite{LST90} &  \\ 
\hline
Hard for factor $< \frac{3}{2}$ & \cite{LST90}  &   \\
\hline
A bound of $T_{opt}+p_{max}$ & \cite{ST93}  &   \\
\hline
$(2-\frac{1}{m})$-approximation & \cite{SV05} & \\
\hline
Integrality gap $\leq 1.95$  & \cite{S12} & $p_{ij}\in \{p_j,\infty\}$  \\
\hline
1.75-approximation & \cite{EKS08}  & $p_{ij}\in \{p_j,\infty\}$, $p_{ij}<\infty$ on $<2$ machines\\  
\hline
$(2-\delta)$-approximation, $\delta>0$ & \cite{CK+15} & $p_{ij}\in \{1,\epsilon\}$ for some $\epsilon>0$\\
\hline
A bound of $T_{opt}+q$ & \cite{VHW14} & $p_{ij}\in \{p,q\}$ for some $p<q$\\
\hline
$\frac{q}{p}$ \& $\frac{3}{2}$-approximation & \cite{P14} & $p_{ij}\in [p,q]$ \& $p_{ij}\in \{1,2,3\}$ \\ 
\hline
\hline
A bound of $T_{opt}+L_{opt}$ & This Work & $p_{ij}\leq T_{opt}$  \\[1ex] 
\hline
A bound of $T_{opt}+\frac {L_{opt}}{\varphi}$ & This Work & $\varphi\geq \frac {L}{T}$, for minimal feasible values of $T$ \& $L$ \\[1ex]
\hline
A bound of $p_{max}+\frac {L_{opt}}{\varphi}$ & This Work & $p_{ij}\in \{p_j,\infty\}$\comment{, $\varphi\geq \frac {\sum_{j\in {\cal J}}p_j}{m\cdot p_{max}}$}\\[1ex]
\hline
\hline
\end{tabular}
\caption{Known results for makespan minimization on unrelated machines.} 
\label{table_unrelated}
\end{table}

\section{Fixed Parameter Algorithms}
Parameterized complexity is a branch of computational complexity theory that focuses on classifying computational problems according to their inherent difficulty with respect to multiple parameters of the input. The complexity of a problem is then measured as a function in those parameters. This allows the classification of NP-hard problems on a finer scale than in the classical setting, where the complexity of a problem is only measured by the number of bits in the input (see, e.g., \cite{DF12}).

Under the assumption that $P \neq NP$, there exist many natural problems that require super-polynomial running time when complexity is measured in terms of the input size only, but that are computable in a time that is polynomial in the input size and exponential (or worse) in a parameter $k$. Hence, if $k$ is fixed at a small value, and the growth of the function over $k$ is relatively small, then such problems can still be considered \textquotedblleft tractable" despite their traditional classification as \textquotedblleft intractable".
\comment{In addition, there exist many natural problems that are hard to approximate within a factor less than some $r>0$ when the complexity of the approximation algorithm is measured in terms of the input size only, but can be approximated within a factor less than $r$ in a time that is polynomial in the input size and exponential or worse in a parameter $k$.}

Some problems can be solved exactly, or approximately, by algorithms that are exponential only in the size of a fixed parameter while polynomial in the size of the input. Such an (approximation) algorithm is called a {\em fixed-parameter tractable (FPT)} (approximation) algorithm, because the problem can be (approximately) solved efficiently for small values of the fixed parameter.

Problems in which some parameter $k$ is fixed are called parameterized problems. A parameterized problem that allows for such an FPT algorithm is said to be a {\em fixed-parameter tractable} problem and belongs to the class FPT.
\comment{
We give some definitions formalizing this concept.
\begin{definition}
A problem is said to be FPT if it can be solved by an algorithm $\cal A $ that runs in time $f(k)\cdot poly(|{\cal I}|)$, for every instance $\cal I$ with parameter $k$, and where $f$ is a function independent of $|{\cal I}|$. The algorithm $\cal A$ is called an FPT algorithm.
\end{definition}

We can similarly define FPT-approximation algorithms.

\begin{definition}
A problem is said to have an $r$-FPT approximation algorithm, if there exists an $r$-approximation algorithm $\cal A$ to the problem, that runs in time $f(k)\cdot poly(|{\cal I}|)$, for every instance $\cal I$ with parameter $k$, and where $f$ is a function independent of $|{\cal I}|$. $f(k)\cdot poly(|{\cal I}|)$, for every instance $\cal I$ with parameter $k$, and where $f$ is a function independent of $|{\cal I}|$. The algorithm $\cal A$ is called an $r$-FPT approximation algorithm.
\end{definition}

\begin{definition}
A family $\{{\cal A}_\epsilon:\epsilon>0\}$, where ${\cal A}\epsilon$ is a $(1+\epsilon)$-FPT approximation algorithm for all $\epsilon>0$, is called a parametrized approximation scheme. 
\end{definition}
}
\subsection{Parametrized Scheduling Problems}
Despite the fundamental nature of scheduling problems, and the clear advantages of fixed-parameter
algorithms, no such algorithms are known for many of the classical scheduling problems. 
One obstacle towards obtaining positive results appears to be that, in contrast to most
problems known to be fixed-parameter tractable, scheduling problems involve many numerical input data
(e.g., job processing times, release dates, job costs), which alone causes many problems to be NP-hard, thus ruling out fixed-parameter algorithms. 

\subsection{Related Work}
While minimizing the makespan on identical and uniform machines admits an EPTAS, see e.g., \cite{AA+98,J10}, the running times of these approximation schemes usually have a bad dependence on $\epsilon$. In addition, these problems are strongly NP-hard \cite{GJ79}, therefore we cannot hope to obtain an FPTAS. These considerations call for finding which scheduling
problems are fixed-parameter tractable (FPT). This amounts to identifying instance-dependent parameters $k$ that allow for algorithms that find optimal solutions in time $f(k)\cdot poly(|{\cal I}|)$, for instances $\cal I$ and some function $f$ depending only on $k$.

As for scheduling on unrelated machines, an FPTAS is known \cite{HS76}, but only when assuming a fixed number of machines. Note that, if the number of machines or the number of processing times are constant, the problem is still NP-hard \cite{LST90}, and thus no FPT-algorithms can exist for these choices of parameters. This motivates us to identify instance-dependent parameters $k$, while assuming an arbitrary number of machines, that allow for better FPT approximation algorithms.   

\comment{Although both fixed-parameter tractability and scheduling are very
well studied areas, no such algorithms are known for many fundamental
scheduling problems. }

Marx \cite{M11} proposed the research direction of scheduling
with rejection, where the parameter $k$ is given with the input for
the scheduling problem, and the solution has to schedule all but $k$
jobs. This direction was explored recently by Mnich and Wiese \cite{MW13}, who presented for the first time a fixed-parameter algorithms for classical scheduling problems such as
makespan minimization, scheduling with job-dependent cost functions and scheduling with rejection.
For the problem of makespan minimization on identical machines, the paper \cite{MW13} presents an FPT algorithm, where the parameter $k$ defines an upper bound on the number of distinct processing times appearing in an instance. For the more general model of unrelated machines, the paper gives an FPT algorithm, using the number of machines and the number of distinct processing times as parameters. 
\newpage
Table \ref{table_FPT} summarizes the known results.

\begin{table}[H]
\centering 
\begin{tabular}{|c|c|c|c|} 
\hline
Result & Parameters & Authors & Model \\ [0.5ex] 
\hline\hline
FPT algorithm & $max_{j}p_j$ & \cite{MW13} & identical machines \\ 
\hline
FPT algorithm & $m$, \#distinct $p_{ij}$ & \cite{MW13}  & unrelated machines  \\
\hline
\hline
FPT approximation scheme & $|\{(i,j):p_{ij}>\epsilon T \}|$, feasible $T$ & This Work & unrelated machines \\[1ex]
\hline
FPT algorithm & treewidth, degree & This Work & graph balancing  \\ [1ex]
\hline
\hline 
\end{tabular}
\caption{Known FPT algorithms for scheduling.} 
\label{table_FPT}
\end{table}

We are not aware of any other work in this area.

\section{Combinatorial Reoptimization}
{\em Reoptimization} problems naturally arise in many real-life scenarios. Indeed,
planned or unanticipated changes occur over time in almost any system. It is then required
to respond to these changes quickly and efficiently. Ideally, the response should maintain high
performance while affecting only a small portion of the system.
Thus, throughout the continuous operation of such a system, it is required to compute
solutions for new problem instances, derived from previous instances. Since the transition
from one solution to another incurs some cost, a natural goal is to have the solution for
the new instance close to the original one (under certain distance measure).

We use the reoptimization model developed by Shachnai et al. \cite{STT12}. In this model, we say that $\cal A$ is an $(r, \rho)$-reapproximation algorithm if it
achieves a $\rho$-approximation for the optimization problem, while incurring a transition cost
that is at most $r$ times the minimum cost required for solving the problem optimally.

\comment{The paper \cite{STT12}  presents reapproximation algorithms for several important
classes of optimization problems. This includes a {\em fully polynomial time reapproximation schemes (FPTRS)} for DP-benevolent problems, reapproximation algorithms for
metric Facility Location problems, and $(1, 1)$-reoptimization algorithm for polynomially
solvable subset-selection problems.

Konstanty et al. \cite{JLR15} considered a variant of a recoloring problem, called the $r$-Color-Fixing. They investigated the problem of finding a proper $r$-coloring of a graph, which is \textquotedblleft most similar" to some given initial solution, i.e. the number of vertices that have to be recolored is minimum possible. They provide a $(1, 1)$-reoptimization algorithm for the problem. More work on reoptimization can be found e.g., in \cite{J15,AE+14}.}

\subsection{Reoptimization in Scheduling}
We consider instances of scheduling
problems that can change dynamically over time. Our goal is to compute
assignments within some guaranteed approximation for the new problem
instances, derived from the previous instances. Since the transition
from one assignment to another incurs some cost (for example, the
cost of pausing the execution of a job on one machine and resuming its execution on another),
an additional goal is to have the solution for the new instance {\em close} to
the original one (under a certain distance measure).
\comment{
Given an scheduling problem $\Pi$, we denote by $R(\Pi)$ the reoptimization version of $\Pi$.

\begin{definition}
An algorithm ${\cal A}$ yields an $(r,\rho)$-reapproximation for $R(\Pi)$, for
$r, \rho \geq1$, if for any instance $\cal I$ of $\Pi$, ${\cal A}$ outputs an assignment of objective value at most $\rho$ times the optimal for $\cal I$, and of total
transition cost at most $r$ times the minimal transition cost to an optimal assignment for $\cal I$.
\end{definition}
}
\subsection{Related Work}
\paragraph{Reoptimization} Shachnai et al. \cite{STT12}  presented reapproximation algorithms for several non-trivial classes of optimization problems. This includes a {\em fully polynomial time reapproximation schemes (FPTRS)} for DP-benevolent problems, reapproximation algorithms for
metric Facility Location problems, and $(1, 1)$-reoptimization algorithm for polynomially
solvable subset-selection problems.

Junosza-Szaniawski et al. \cite{JLR15} considered a variant of a recoloring problem, called the $r$-Color-Fixing. They investigated the problem of finding a proper $r$-coloring of a graph, which is \textquotedblleft most similar" to some given initial solution, i.e. the number of vertices that have to be recolored is minimum possible. They provide a $(1, 1)$-reoptimization algorithm for the problem. More work on reoptimization can be found e.g., in \cite{J15,AE+14}.

\paragraph{Reoptimization in Scheduling}
Baram and Tamir \cite{BT14} considered the problem of scheduling on identical machines with the objective of minimizing the total flow time, i.e., minimizing $\sum_{j\in {\cal J}} C_j$, where $C_j$ is the completion time of job $j$. \comment{ They study two different problems: (i) achieving an optimal solution
using the minimal possible transition cost, and (ii) achieving the best
possible schedule using a given limited budget for the transition.} They presented an algorithm that yields an optimal solution using the minimal possible transition cost, and an algorithm that outputs the best possible schedule, using a given limited budget for the transition, for several classes
of instances.

Bender et al. \cite{BF+13} focused on a scheduling problem where each job is unit-sized and has a time window in which it can be executed. Jobs are dynamically added and removed from
the system. They presented an algorithm that reschedules only $O(min\{log^*n , log^* \Delta \})$ jobs for each job that is inserted or deleted from the system, where $n$ is the number of active
jobs, and $\Delta$ is the size of the largest window.

\section{Main Results}

Our first contribution is in the study of makespan minimization on
unrelated machines, which leads to approximation algorithms with performance
guarantees better that $2$, the best known bound for general instances.
In particular, for the subclass of fully-feasible instances, we present (in
Chapter 2) an LP-based algorithm that achieves makespan at most $T_{opt}+L_{opt}$, where $T_{opt}$ and $L_{opt}$ are the minimal makespan and minimal average machine load for this makespan, respectively. This result is better than twice the minimal makespan for instances naturally arising in real-life applications. It also improves the minimal makespan plus the maximum processing time of a job, for instances where the average machine load is smaller than the maximum processing time of any job. We show that our algorithm is {\em robust} in the sense that it achieves an improved makespan also for instances that are {\em almost} fully-feasible. In such instances, each job may exceed the minimal makespan on a {\em small} number of machines. Formally, we define the {\em feasibility parameter} of a general instance $\cal I$, denoted $\varphi(\cal I)$, as the minimal fraction of machines on which a job has a processing time at most $T_{opt}$, i.e., $\varphi({\cal I})=\frac{min_{j\in {\cal J}} |\{i\in {\cal M}:p_{ij}\leq T_{opt}\}| }{m}$. We present an algorithm that yields, for instances with large {\em enough} feasibility parameter, a schedule of makespan at most $T_{opt}+\frac {L_{opt}}{\varphi(\cal I)}$.
For instances of the restricted assignment problem, i.e., for instances with processing times $p_{ij}\in \{p_j,\infty\}$, we show that a bound of $p_{max}+\frac {L_{opt}}{\varphi(\cal I)}$ is obtained by an efficient and simple combinatorial algorithm, where $p_{max}$ is the largest processing time of any job in the instance.

We further study the power of parameterization and present an FPT approximation scheme, i.e., a $(1+\epsilon)$-FPT approximation algorithm, for makespan minimization on unrelated machines parametrized by the number of machine-job pairs, $(i,j)\in {\cal M}\times \cal J$, such that $p_{ij}> \epsilon T$, for some makespan candidate $T$. We also show that the graph-balancing problem, parameterized by treewidth and the maximum degree of the graph, is in FPT. These results are presented in Chapter 3.

Our third contribution is reapproximation algorithms for the reoptimization variants of makespan minimization on identical and uniform machines, which are studied here for the first time. Specifically, we develop $(1,1+\epsilon)$-reapproximation algorithms, namely, algorithms that achieve a
ratio of $(1+\epsilon)$ to the minimum makespan, and the minimum transition
cost, in both the identical machines and the uniform machines models, where 
transition costs can take values in $\{0,1\}$. For the uniform case, we assume that the ratio between the highest and the lowest machine speeds is bounded by some constant. Thus, our algorithms achieve the best possible ratio with respect to the makespan objective in these models. For the unrelated machines model, we note that an algorithm of Shmoys and Tardos [ST93] can be used to obtain a
$(1,2)$-reapproximation, thus matching the best known bound for makespan
minimization also in this model.
We summarize the results for reoptimization in scheduling in Table \ref{table_reopt}. The
results are given in Chapter 4.

\begin{table}[H]
\centering 
\begin{tabular}{|c|c|c|c|} 
\hline
Result & costs & Authors & Model \\ [0.5ex] 
\hline\hline
$(1,2)$-reapproximation algorithm & arbitrary & \cite{ST93} &  unrelated machines \\ 
\hline
\hline
$(1,1+\epsilon)$-reapproximation algorithm & \{0,1\} & This Work  & identical machines  \\[1ex]
\hline
$(1,1+\epsilon)$-reapproximation algorithm & \{0,1\} & This Work  & uniform machines with $\frac{s_1}{s_m}\leq b$  \\[1ex]
\hline
\hline
\end{tabular}
\caption{Our contribution for reoptimization in scheduling.} 
\label{table_reopt}
\end{table}

\chapter{Makespan Minimization for Fully-Feasible Instances}

\section{Preliminaries}
\label{sec:prel}

Consider a scheduling instance $\cal I=(\cal J,\cal M)$, consisting of a set of $m$ machines $\cal M$ and a set of $n$ jobs $\cal J$ with non-negative integers $p_{ij}$ denoting the processing time of job $j\in \cal J$ on machine $i\in \cal M$. An \emph{assignment} of the jobs to the machines is a bijection
$\sigma:{\cJ} \rightarrow {\cM}$ where $\sigma(j)=i$ if and only if job $j$ is
assigned to machine $i$. For any assignment $\sigma$, the \emph{load} on machine $i$ under assignment $\sigma$, denoted as $load_{\sigma}(i)$, is the sum of processing times for the jobs that were assigned to machine $i$. Thus, $load_{\sigma}(i)=\sum_{j\in {\cJ}:\alpha(j)=i}p_{ij}$.
The \emph{makespan} of an assignment $\sigma$, denoted by $T(\sigma)$, is the maximum load over all the machines. Thus, $T(\sigma)=max_{i\in \cal M}load_\sigma (i)$.
The \emph{average machine load} of an assignment $\sigma$, denoted by $L(\sigma)$, is given by $L(\sigma)=\frac{\sum_{i\in {\cM}}load_\sigma (i)}{m}$.
Given an instance $\cal I=(\cal J,\cal M)$, we denote by $T_{opt}$ the optimal makespan, i.e., $T_{opt}=min_{\sigma:\cal M \rightarrow \cal J}T(\sigma)$, and we denote by $L_{opt}$ the minimum average machine load for an optimal assignment, i.e., $L_{opt}=min_{\sigma^{*}:{\cal M} \rightarrow {\cal J}, T(\sigma^{*})=T_{opt}}L(\sigma^{*})$.

Given an instance $\cal I=(\cal J,\cal M)$, we say that a job $j$ is {\em feasible} on machine $i$, if and only if $p_{ij}\leq T_{opt}$. The \emph{feasibility parameter} of $\cal I$, denoted $\varphi(\cal I)$, is the minimal fraction of feasible machines for any given job, i.e., $\varphi({\cal I})=min_{j\in {\cJ}} \frac{|\{i\in {\cal M}: j \mbox{ is feasible on } i\}|}{m}$. Thus, every job $j\in \cal J$ is feasible on at least $\varphi({\cal I})\cdot m$ machines. In this terms, an instance $\cal I$ is {\em fully-feasible} if and only if $\varphi({\cal I})=1$. We often omit $\cal I$ in the notation if it is clear
from the context, and refer to the feasibility parameter as $\varphi$.

Given an instance $\cal I=(\cal J,\cal M)$, an assignment $\sigma:{\cal J} \rightarrow {\cal M}$, two positive integers $L$ and $T$ and some real number $\gamma\geq1$, we denote by $Bad(\sigma,\gamma)\subseteq {\cal M}$ the subset of machines $i$ with $load_\sigma(i)>T+\gamma\cdot L$, by
$Good(\sigma,\gamma)\subseteq \cal M$ the subset of machines $i$ with $load_\sigma(i)\leq \gamma\cdot L$, and for all $j\in \cal J$, by $Good_{j}(\sigma,\gamma)\subseteq Good(\sigma,\gamma)$
the set of machines from $Good(\sigma,\gamma)$ that are also legal for job $j$. 
For every $i\in \cal M$ we denote by $j_{max}^{i}(\sigma)=argmax\left\{ p_{ij}:\,\sigma(j)=i\right\}$ the job with the longest processing time, assigned, by $\sigma$, on machine $i$.

\comment{We denote by $G_{\sigma\gamma}$ the bipartite
graph $\left(\left(Bad(\sigma,\gamma),Good(\sigma,\gamma)\right),E(\sigma)\right)$
where $E(\mathbf{P},\alpha)$ consists of all edges $(i,i^{'})$ where
$i^{'}$ is a legal machine for $j_{max}^{i}(\mathbf{P},\alpha)$.
We say that machine $i$ is\emph{ good }if $i\in Good(\mathbf{P},\alpha,\gamma)$.
We say that machine $i$ is\emph{ bad }if $i\in Bad(\mathbf{P},\alpha,\gamma)$.
Machine $i$ is\emph{ good for job $j$ }if $i\in Good_{j}(\mathbf{P},\alpha,\gamma)$.
}

\section{Approximation Algorithm for Fully-Feasible Instances}
\label{sec:general}

Given positive integers $L$ and $T$, let $x_{ij}$ be an indicator to the assignment of job $j$ on machine $i$.
Consider the following linear program.

\begin{center}
\begin{minipage}{0.95\textwidth}
\begin{eqnarray*}
LP(T,L): & \displaystyle{\frac{1}{m} \sum_{i=1}^{m}\sum_{j=1}^{n}p_{ij}x_{ij}\leq L} \hbox{~~~~~~~~~~~~~~~~~~~~~~~~~~~~~~~~ }\\
& \displaystyle{\sum_{i=1}^{m}x_{ij}=1}, \hbox{~~~~~~~~~~ for \ensuremath{j=1,...,n} ~~~~~~~~ }  \\
& \displaystyle{\sum_{j=1}^{n}p_{ij}x_{ij}\leq T}, \hbox{~~~~~ for \ensuremath{i=1,...,m} ~~~~~~~~ } \\
& ~~~~~x_{ij}\geq0, \hbox{~~~~~~~~~~~~~~~~ for  \ensuremath{i=1,...m\,,j=1,...,n} } \\
&&
\end{eqnarray*}
\end{minipage}
\end{center}

One can see that integer solutions to the above LP are in one to one correspondence with assignments $\sigma:\cal M \rightarrow \cal J$ of average machine load at most $L$ and makespan at most $T$.
\comment{Using the rounding technique, also used in \cite{ST93} for the generalized assignment problem, where every job also incurs a cost of $c_ij$ if scheduled on machine $i$, we can round any fractional solution $\bar x$ to $LP(T,L)$, to an assignment $\sigma:\cal J \rightarrow \cal M$, of average machine }
 
\comment{In \cite{ST93}, Shmoys and Tardos present a polynomial time algorithm, that given values $L$ and $T$, if $LP(T,L)$ is feasible, yields an assignment $\sigma$ with cost at most $mL$ and $load_{\sigma}(i)=T+max\{p_{ij}:\sigma(j)=i\}$, $\forall i $
Recall the \emph{generalized assignment problem}, a generalization of the problem of scheduling on unrelated machines, where each job $j\in \cal J$ also incurs a cost $c_{ij}$ on machine $i\in \cal M$. The best known result for the generalized assignment problem is due to \cite{ST93}. They presented
a polynomial time algorithm that, given values $C$ and $T$, finds
a schedule of cost at most $C$ and makespan at most $2T$, if a schedule
of cost $C$ and makespan $T$ exists. This algorithm yields a schedule makespan not larger than twice the optimum with optimal job cost. 
The next theorem then derives when fixing $c_{ij}=p_{ij}$, $\forall i\in \cal M, j\in \cal J$.}

\begin{theorem}
\label{initial assignment thm}
If $LP(T,L)$ is feasible for some $L\leq T$, then there is a polynomial-time algorithm that yields a schedule $\sigma$ with
\begin{enumerate}
\item $L(\sigma)\leq L$
\item $load_{\sigma}(i)\leq T+max\{p_{ij}:\sigma(j)=i\}$, $\forall i\in \cal M$
\end{enumerate}
\end{theorem}

\begin{proof}
Let $\bar x = (x_{ij}:i\in \cal M,j\in \cal J)$ be a fractional solution to $LP(T,L)$.
We round $\bar x$ to an integer solution using the following classic rounding technique, also used in \cite{ST93}.  
Let $k_{i}=\left\lceil \sum_{j\in {\cal J}}x_{ij}\right\rceil $. Each machine is partitioned into $k_{i}$ sub-machines $v_{i,s}$, $s=1,...,k_{i}$. \comment{The rounding is done by finding
a minimum-cost perfect matching between all jobs and all sub-machines.}

An edge weighted, bipartite graph $B=(W,V;E)$ is constructed, where $W=\left\{ w_{j}:j=1,...,n\right\} $
representing the jobs and $V=\left\{ v_{is}:i=1,..,m,\, s=1,...,k_{i}\right\} $ representing the sub-machines. The edges are constructed in the following way.

Consider the nodes $v_{is}$ as bins of unit capacity and the nodes $w_j$, as pieces of size $x_{ij}$. 
For every machine $i\in \cal M$, consider the nodes in non-increasing order of the processing time of the corresponding job, on machine $i$. For convenience, assume $p_{i1}\geq p_{i2}\geq ...\geq p_{in}$.
An edge $(w_{j},v_{is})$ with cost $p_{ij}$ is constructed if and only if a positive fraction of $x_{ij}$ is packed in the bin $v_{is}$. The packing of the bins (and construction of the edges) is done in the following way. The bins $v_{i1},...,v_{ik_{i}}$ are packed one by one, with the pieces in the order $p_{i1},p_{i2},...,p_{in}$. While $v_{is}$ is not totally packed, (otherwise we consider the next bin) we continue packing the next piece such that if its size, $x_{ij}$, fits $v_{is}$ (without causing an overflow) it is packed to $v_{is}$, else, if it causes an overflow, only a fraction $\beta> 0$ of $x_{ij}$ is packed to $v_{is}$, consuming all the remaining volume of $v_{is}$, and the remaining part of $(1-\beta)x_{ij}$ is packed in $v_{i,s+1}$.
Figure \ref{rounding} gives a pictorial example of this construction.

The rounding is done by finding a minimum-cost integer matching $M\in E$ that matches all job nodes, and for every edge $(w_{j},v_{is})\in M$, set $\sigma(j)=i$ i.e., assign job $j$ on machine $i$.

By taking a minimum-cost integer matching it is guaranteed that $\frac{1}{m} \sum_{i=1}^{m}\sum_{j:\sigma(j)=i}p_{ij}\leq L$, or $L(\sigma)\leq L$.
By the construction of the graph it is guaranteed that the load on machine $i$, for every $i\in \cal M$, is at most $max\{p_{ij}:\sigma(j)=i\}+\sum_{s=2}^{k_i} p_{is}^{max}$, where $p_{is}^{max}=max\{p_{ij}:(w_{j},v_{is})\in E\}$.   
Since $\sum_{s=2}^{k_i} p_{is}^{max}\leq T$ (for a detailed proof, see \cite{ST93}), we get that $load_\sigma (i)\leq max\{p_{ij}:\sigma(j)=i\}+T$ for all $i\in \cal M$.
\end{proof}

\begin{figure}
\begin{center}
{\epsfig{file=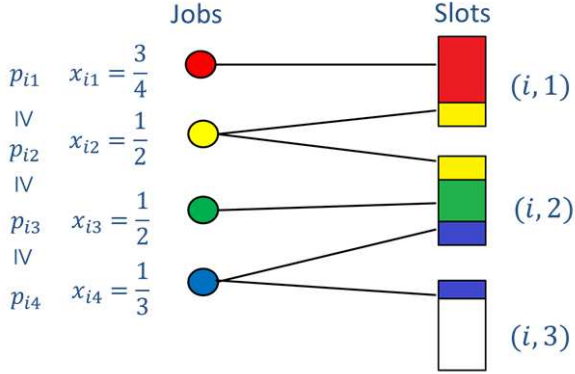,width=4.5in}}
\caption{\small Example of the construction in the rounding}
\label{rounding}
\end{center}
\end{figure}

This result will be helpful in our approximation algorithm.

\subsection{Approximation Algorithm}

Recall that an instance $\cal I$ is {\em fully-feasible} if and only if $p_{ij}\leq T_{opt}$ for every job $j\in \cal J$ and machine $i\in \cal M$.
Although fully-feasible instances are hard to identify, we give a polynomial-time algorithm that yields for such instances an assignment whose makespan is better than twice the optimal makespan, the best known ratio for general instances. 

Our main result is as follows.
\begin{theorem}
\label{fully-Feasible Thm}
There is a polynomial-time algorithm that yields for fully-feasible instances an assignment of makespan at most $T_{opt}+L_{opt}$.
\end{theorem}

Note that always $L_{opt}\leq T_{opt}$. $L_{opt}=T_{opt}$ occurs only in case where {\em every} optimal assignment is perfectly balanced (i.e., the load on all the machines equals $T_{opt}$). Thus, in the typical case we get $L_{opt}<T_{opt}$, which means that in the typical case our algorithm guaranteed a bound that is strictly better that twice the optimum. 

The following theorem shows that the problem remains hard already when considering only the class of fully-feasible instances.

\begin{theorem}
Makespan minimization on unrelated machines is hard to approximate within factor $<\frac{4}{3}$ , already in the class of fully-feasible instances.
\end{theorem}

The proof of the hardness result follows from a reduction from $3$-dimensional matching as in \cite{LST90}.

We prove a stronger result, that also shows that our result is robust under small violations of the feasibility constraints. Namely, we show that we can get better bounds also when our instance is not fully-feasible but has the property that every job can be non-feasible on some {\em small} fraction of the machines. For this, we will need some definitions.

\begin{definition}
Job $j\in \cal J$ is feasible on machine $i\in \cal M$ if $p_{ij}\leq T_{opt}$
\end{definition}

\begin{definition}
The feasibility parameter of instance $\cal I$ is defined as
\begin{equation*}
\quad \varphi({\cal I})=min_{j\in {\cJ}} \frac{|\{i\in {\cal M}: j \mbox{ is feasible on } i\}|}{m}
\end{equation*}
\end{definition}

In these terms, we have that every job $j\in \cal J$ is feasible on at least $\varphi({\cal I})\cdot m$ machines. Moreover, an instance $\cal I$ is fully-feasible if and only if $\varphi({\cal I})=1$.

$  $

Given an instance $\cal I$, we can fix $L$ and $T$ to be the minimal values such that $LP(T,L)$ is feasible in polynomial-time in the size of the input $|\cal I|$, using binary search on some feasible regions for $L$ and $T$.
At the end of this operation, $L$ and $T$ will satisfy $L\leq T$, $L\leq L_{opt}$ and $T\leq T_{opt}$.

\begin{theorem}
\label{general main theorem}
There is a polynomial-time algorithm that, given the values $L\leq T$, yields for instances $\cal I$ with feasibility parameter $\varphi({\cal I}) \geq \frac{L}{T}$, an assignment of makespan at most $T_{opt}+\frac{L_{opt}}{\varphi(\cal I)}$
\end{theorem}

\begin{corollary}
The statement of theorem \ref{fully-Feasible Thm} follows directly from \ref{general main theorem}, since the feasibility parameter of fully-feasible instances is $1$.
\end{corollary}

Given a general scheduling instance $\cal I$, its optimal makespan, $T_{opt}$, nor its feasibility parameter, $\varphi(\cal I)$, cannot be computed in polynomial-time, unless $P=NP$. However, let $\ell$ be the number of distinct processing times $p_{ij}$ of an instance $\cal I$ and w.l.o.g assume that $p_1\leq p_2 \leq ... \leq p_\ell$ are the distinct processing times of $\cal I$.
If we denote $p_{\ell+1}=\infty$, then we see that $T_{opt}$ can belong to exactly one of the regions $[p_t,p_{t+1})$ for some $t\in\{1,...,\ell\}$. 
If $T_{opt}\in [p_t,p_{t+1})$ for some $t\in\{1,...,\ell\}$, then job $j\in \cal J$ is feasible on some machine $i\in \cal M$ if and only if $p_{ij}\leq p_t$. Therefore, $T_{opt}\in [p_t,p_{t+1})$ if and only if $\varphi({\cal I})=\varphi_t\triangleq min_{j\in {\cJ}} \frac{|\{i\in {\cal M}: p_{ij}\leq p_t \}|}{m} $.

$   $

We prove the theorem by describing an algorithm that admits the desired bound on the makespan, for instances with {\em large} enough feasibility parameter.
The first step of the algorithm is to find the minimal $T$ and $L$ such that $LP(T,L)$ is feasible. 
Next, it obtains a fractional solution to $LP(T,L)$ and rounds it to obtain an initial assignment $\sigma$, such that $L(\sigma)\leq L$ and $load_{\sigma}(i)\leq T+max\{p_{ij}:\sigma(j)=i\}$, $\forall i\in \cal M$, as in Theorem \ref{initial assignment thm}.
Then, for each guess of the optimal makespan region and the corresponding feasibility parameter, it tries to fix the initial assignment to achieve a new assignment of makespan at most $T_{opt}+\frac{L_{opt}}{\varphi}$ (and unless the feasibility parameter is too large, it will succeed for the right guess). This is done by balancing the initial assignment such as to reduce the load of the overloaded machines to meet the desired makespan.  

We first prove the following lemmas.

\begin{lemma}
\label{enough_good_machines}
Let $\sigma$ be an assignment for some instance $\cal I=(\cal J,\cal M)$ such that $L(\sigma)\leq L$, let $T\geq L $ and let $\gamma\geq1$.
Denote  $k=|Bad(\sigma,\gamma)|$. Then
\begin{enumerate}
\item $k<\frac{m}{\gamma+1}$.
\item  $|Good(\sigma,\gamma)|>\left(1-\frac{1}{\gamma}\right)\cdot m+\frac{k}{\gamma}\cdot\frac{T}{L}$.
\end{enumerate}
\end{lemma}

\begin{proof}
Each machine $i\in Bad(\sigma,\gamma)$ has load greater than $T+\gamma\cdot L$,
therefore $\sum_{i\in {\cal M}}load_\sigma(i) > k\cdot(T+\gamma\cdot L)$.
\begin{enumerate}
\item Assume that $k\geq\frac{m}{\gamma+1}$, then

\[
\begin{array}{ll}
\sum_{i\in M}load_\sigma(i) & >  k(T+\gamma L)\\
 & \geq \frac{m}{\gamma+1}(T+\gamma L)\\
 & = \frac{m}{\gamma+1}(T+\gamma L)\\
& {\geq}
 \frac{m}{\gamma+1}(L+\gamma L)\\
 & =  m\cdot L
\end{array}
\]

The last inequality follows from the fact that $T \geq L$.
Hence, the average machine load is greater than $L$, a contradiction.
It follows that $k<\frac{m}{\gamma+1}$.

\item Let $\left|Good(\sigma,\gamma)\right|=l$.
Then, there are $m-k-l$ machines having loads greater than $\gamma L$.
Assume that $l\leq\left(1-\frac{1}{\gamma}\right)m+\frac{k}{\gamma}\cdot\frac{T}{L}$, then

\[
\begin{array}{ll}
\sum_{i\in M}load_\sigma(i) & >  k(T+\gamma\cdot L)+\left(m-l-k\right)\gamma L\\
 & = kT+\left(m-l\right)\gamma L\\
 & \geq  kT+\left(m-\left(1-\frac{1}{\gamma}\right)m+\frac{k}{\gamma}\cdot\frac{T}{L}\right)\gamma L\\
 & \geq kT+\left(\frac{m}{\gamma}+\frac{k}{\gamma}\cdot\frac{T}{L}\right)\gamma L\\
 & =  kT+\left(mL+kT\right)\\
 & \geq mL
\end{array}
\]

Hence, the average machine load is greater than $L$, a contradiction.
It follows that $\left|Good(\sigma,\gamma)\right|\geq\left(1-\frac{1}{\gamma}\right)\cdot m+\frac{k}{\gamma}
\cdot\frac{T}{L}$.
\end{enumerate}
\end{proof}

\begin{lemma}
Let $\cal I=(\cal J,\cal M) $ be an instance with feasibility parameter $\varphi$. Let $\sigma$ be an assignment for $\cal I$ with average machine load $L$ and let $T\geq L$. If $\varphi \geq\frac{L}{T}$ then for every subset $A\subseteq Bad(\sigma,\frac{1}{\varphi})$, $\left|N\left(A\right)\right|\geq\left|A\right|$, where $N\left(A\right)$ is the set of neighbors of $A$ in $G_{\sigma,\frac{1}{\varphi}}$.
\end{lemma}
\begin{proof}
Let $\left|Bad(\sigma,\frac{1}{\varphi})\right|=k$. Since the number
of illegal machines for any job $j$ is at most $(1-\varphi)m$, the number of good machines for job $j$ is at least the number of good machines minus its illegal machines (the worst case where all illegal machines for job $j$ are contained in $Good(\sigma,\frac{1}{\varphi})$).
Together with Lemma \ref{enough_good_machines} we have

\[
\begin{array}{ll}
\left|Good_{j}(\sigma,\frac{1}{\varphi})\right| & \geq  \left|Good(\sigma,\frac{1}{\varphi})\right|-(1-\varphi)m\\
 & >  \left(1-\frac{1}{\varphi}\right)\cdot m+\frac{k}{\left(\frac{1}{\varphi}\right)}\cdot\frac{T}{L}-(1-\varphi)m\\
  & = \varphi\cdot k\frac{T}{L}\\
  & \geq k
\end{array}
\]

The last inequality follows from the fact that $\varphi\geq \frac{L}{T}$.
Now, let $A\subseteq Bad(\alpha,\frac{1}{\varphi})$. Then $\left|A\right|\leq\left|Bad(\sigma,\frac{1}{\varphi})\right|=k$.

Recall that the set of neighbors of $A$ is the set of machines that
are good for all the jobs $j_{max}^{i}$, $i\in A$, i.e., $N\left(A\right)=\cup_{i\in A}Good_{j_{max}^{i}}(\sigma,\frac{1}{\varphi})\subseteq Good(\sigma,\frac{1}{\varphi})$.
Obviously $\left|N\left(A\right)\right|=\left|\cup_{i\in A}Good_{j_{max}^{i}}(\sigma,\frac{1}{\varphi})\right|\geq\left|Good_{j_{max}^{i}}(\sigma,\frac{1}{\varphi})\right|$ for some $i\in A$. It follows from the above that $\left|N\left(A\right)\right|\geq k$.

Since $\left|A\right|\leq k$ we have that $\left|N\left(A\right)\right|\geq\left|A\right|$.
\end{proof}

By Hall's Theorem \cite{H35}, there exist a perfect matching in
$G_{\sigma,\frac{1}{\varphi}}$ if and only if  for every $A\subseteq Bad(\sigma,\frac{1}{\varphi})$,
$\left|N\left(A\right)\right|\geq\left|A\right|.$ Thus, we have
\begin{corollary}
\label{lemma:hall}
There exists a perfect matching in $G_{\sigma,\frac{1}{\varphi}}$.
\end{corollary}

By the above discussion, we can modify the initial assignment $\sigma$, by finding a perfect
matching in $G_{\sigma,\frac{1}{\varphi}}$ and then transferring jobs
from bad machines to their matching good machines. We describe this formally in algorithm ${\cal A}_{UM}$.

\begin{algorithm}[H]
\caption{${\cal A}_{UM}$}
\begin{enumerate}
\item Use binary search to find the minimal $T$, such that $LP(T,T)$ is feasible. Next, with that $T$ fixed, search for the minimal $L$ such that $LP(T,L)$ is feasible.
\item Solve the linear relaxation $LP(T,L)$.
\item Round the solution to obtain an integral assignment $\sigma$ using a rounding technique as
given in Theorem \ref{initial assignment thm}. \label{alg:step3}
\item For every $t=1,...,\ell$, where $\ell$ is the number of distinct processing times of $\cal I$, guess that $T_{opt}\in [p_t,p_{t+1})$ and that $\varphi({\cal I})=\varphi_t$. 
\begin{enumerate}
\item If $\varphi_t < \frac{L}{T}$, continue.
\item Otherwise, construct the bipartite graph $G_{\sigma,\frac{1}{\varphi_t}}$
and find a perfect matching of size $\left|Bad(\sigma, \frac {1}{\varphi_t})\right|$, if one exists. If not, continue.
\item Obtain a resulting assignment $\sigma^{'}$ from $\sigma$ by transferring
the longest job, $j_{max}^{i}$ from each machine $i\in Bad(\sigma, \frac {1}{\varphi_t})$,
to its matching machine $i^{'}\in Good(\sigma, \frac {1}{\varphi_t})$.
\end{enumerate}
\comment{\item If no perfect matching was found for any $t=1,...,\ell$, then apply the $2$-approximation algorithm of \cite{LST90}, and return the resulting assignment. }
\item Return the assignment $\sigma^{'}$ with minimal makespan. 
\end{enumerate}
\end{algorithm}

\noindent
{\bf Proof of Theorem \ref{general main theorem}.}
We show that the assignment output by Algorithm ${\cal A}_{UM}$
satisfies the statement of the theorem. Consider a general instance $\cal I$.
By performing binary search on a feasible bounded region of the optimal makespan we can find the minimal $T$ for which $LP(T,T)$ is feasible, and then by performing binary search on a feasible bounded region of the optimal average machine load we can find the minimal $L$ for which $LP(T,L)$ is feasible. These integers satisfy that $T\leq T_{opt} $, $L\leq L_{opt}$ and $L\leq T$.

By Theorem \ref{initial assignment thm}, since $LP(T,L)$ is feasible then Step \ref{alg:step3} is guaranteed to generate an assignment $\sigma$ of  $L(\sigma)\leq L$ and $load_{\sigma}\leq T+max\{p_{ij}:\sigma(j)=i \}$, for all $i\in \cal M$.

Let $\varphi \geq\frac{L}{T}$ be the feasibility parameter of $\cal I$. Then for some $t=1,...,\ell$, where $\ell$ is the number of distinct processing times in $\cal I$, $\varphi_t=\varphi$. Then, by Corollary \ref{lemma:hall}, there exists a perfect matching in $G_{\sigma,\frac{1}{\varphi_t}}$ and we will find it in Step $4(b)$ in ${\cal A}_{UM}$.
Let $k=\left|Bad(\sigma, \frac{{1}}{\varphi_t})\right| $ and let $M=\left\{ (i_{b_{1}},i_{g_{1}}),...,(i_{b_{k}},i_{g_{k}})\right\} $ be a perfect matching in $G_{\sigma,\frac{1}{\varphi_t}}$.

For any machine $i=1,...,m$, $load_\sigma (i)\leq T+max\{p_{ij}:\sigma(j)=i\}$. Let $j_{max}^{i}$ be the largest job processed by $\sigma$ on machine $i$. Then, $p_{i,j_{max}^{i}}=max\{p_{ij}:\sigma(j)=i\}$, thus removing $j_{max}^{i}$ from machine $i$
guarantees that the new load of machine $i$ will be at most $T$.

As for the good machines, if $i$ is a good machine for job $j$ then
$p_{ij}\leq p_t$.  Therefore, transferring $j$ to $i$ will increase the
load of $i$ by at most $p_t$ which is at most $T_{opt}$ (since $T_{opt}\in [p_t,p_{t+1})$). Since the load of a good machine is at most $\frac{L}{\varphi_t}$, we have that after such a job transfer
the load will be at most $T_{opt}+\frac{L}{\varphi_t}$.

The load on the rest of the machines stays unchained, i.e., $load_{\sigma}(i)\leq T+\frac {L}{\varphi_t}$, for all $i\notin Bad(\sigma,\frac {1}{\varphi_t})\cup Good(\sigma,\frac {1}{\varphi_t})$.

Thus, by performing the large-jobs transfers for all pairs $(i_{b_{s}},i_{g_{s}})\in M$, $s=1,...,k$,
we obtain a new assignment $\sigma^{'}$, with  $load_{\sigma^{'}}(i)\leq max\{T, T_{opt}+\frac {L}{\varphi_t}, T+\frac {L}{\varphi_t} \}$ for all $i\in \cal M$, which is at most $T_{opt}+\frac{L_{opt}}{\varphi_t}$, since $L\leq L_{opt}$ and $T\leq T_{opt}$.
$\qed$

\begin{theorem}
The complexity of ${\cal A}_{UM}$ is  $O((nm)^{\frac{7}{2}}\cdot log^{2}(\sum_{(i,j)\in{\cal M}\times {\cal J} }p_{ij}))$.
\end{theorem}
\begin{proof}
We will show that Step 1 in $\cal A_{UM}$ is the bottle-neck of the algorithm. In this step we perform a binary search on the feasible regions of $T$ and $L$ while solving $LP(T,L)$.
Since the feasible region for both $T$ and $L$ is $[0,\sum_{(i,j)\in{\cal M}\times {\cal J} }p_{ij}]$, we solve the LP $O(log(\sum_{(i,j)\in{\cal M}\times {\cal J} }p_{ij}))$ times. Solving the LP can be done in time $O((nm)^{\frac{7}{2}}\cdot log(\sum_{(i,j)\in{\cal M}\times {\cal J} }p_{ij}))$  \cite{K84}, so overall this operation runs in $O((nm)^{\frac{7}{2}}\cdot log^{2}(\sum_{(i,j)\in{\cal M}\times {\cal J} }p_{ij}))$.

The number of vertices in the bipartite graph $G_{\sigma, \frac{1}{\varphi_t}}$ equals to the number of bad and good machines, which is at most $m$. From \ref{enough_good_machines}, the number of bad machines is at most $\frac{m}{\frac{1}{\varphi}+1}$, thus the number of edges in the bipartite graph is at most $\sum_{i\in Bad(\sigma, \frac{1}{\varphi})}\varphi m=\frac{\varphi}{\frac{1}{\varphi}+1}\leq \frac{\varphi}{2}\leq \frac{1}{2}$ (the last two inequalities are due to $\varphi\leq 1$). Therefore finding a perfect matching in $G_{\sigma, \frac{1}{\varphi_t}}$ can be done in time $O(\sqrt{|V_{\sigma, \frac{1}{\varphi_t}}|}\cdot |E_{\sigma, \frac{1}{\varphi_t}}|)=O(\sqrt{m})$ \cite{MV80}. In ${\cal A}_{UM}$, we find a perfect matching for every $t=1,..,\ell$, where $\ell\leq n\cdot m$ is the number of distinct processing times. Thus, the complexity of Step 4 sums to $O(nm\sqrt{m})$. 

Step 2 is done in $O((nm)^{\frac{7}{2}}\cdot log(\sum_{(i,j)\in{\cal M}\times {\cal J} }p_{ij}))$ time and Step 3 is done in $O(n^3(m+n)^3)$ \cite{EK72}.

Therefore, the overall complexity of the algorithm is $O((nm)^{\frac{7}{2}}\cdot log^{2}(\sum_{(i,j)\in{\cal M}\times {\cal J} }p_{ij}))$.

\end{proof}

\newpage

\section{A Better Bound for the Restricted Assignment Problem}
\label{sec:restricted}
In this section we consider the restricted version of our problem, where $p_{ij}\in\left\{ p_{j},\infty\right\}$, for each job $j\in \cal J$, and each machine $i\in \cal M$.  This subclass is NP-hard to approximate within a factor better than $\frac{3}{2}$, which is also the best known
lower bound for the general version \cite{LST90}. In the restricted version, any job $j$ with processing time $p_{ij}<\infty$ is feasible on machine $i$, since $p_j\leq T_{opt}$ for every $j\in \cal J$. Hence, the feasibility parameter of a restricted instance $\cal I$ is exactly $\varphi({\cal I})=\frac {min_{j\in {\cal J}}|\{i\in {\cal M}:p_{ij}<\infty\}|}{m}$, which can be computed efficiently. Fully-feasible restricted instances correspond to the identical machines instances, hence we do not consider especially fully-feasible instances in this case. Also, we say that an assignment $\sigma:{\cal M}\rightarrow {\cal J}$ is {\em feasible} if $p_{ij}<\infty$ for every machine $i$ and job $j$ such that $\sigma(j)=i$.

For this variant, we show that a better bound than in Theorem \ref{general main theorem} can be achieved by a much simpler and more efficient combinatorial algorithm, and for every feasibility parameter. Denote by $p_{max}=max_{j\in {\cal J}}p_j$ the largest processing time of some restricted instance ${\cal I}$. Gairing et al. \cite{GL+04} presented a $(2-\frac{1}{p_{max}})$-approximation algorithm for the restricted assignment problem. Using techniques from \cite{GL+04}, we obtain an approximation algorithm which yields an assignment of makespan at most $p_{max}+\frac{L_{opt}}{\varphi}$, where $\varphi$, is the feasibility parameter of the instance.

\paragraph{Overview of the Algorithm of Gairing et al.}
\label{sec:Gairing overview}

We describe below the procedure ${\cal UBF}$, used in \cite{GL+04}.
Let $\cal I$ be an instance for the restricted assignment problem. Let $\Delta$ be an integer that will be determined by binary search, to be a lower bound on $T_{opt}$.

Let $\sigma$ be a feasible assignment and let $G_\sigma=(W\cup V, E_\sigma)$ be a directed bipartite graph where $W=\left\{ w_{j}:j\in {\cal J}\right\}$ consists of the job nodes, and $V=\left\{ v_{i}:i\in {\cal M}\right\}$ consists of the machine nodes. For any job node $w_j$ and any machine node $v_i$, if $\sigma(j)=i$ there is an arc in $E_\sigma$ oriented from $v_i$ to $w_j$; if $\sigma(j)\neq i$ and $j$ is feasible on machine $i$, then there is an arc in $E_\sigma$ oriented from $w_j$ to $v_i$.

Given a feasible assignment $\sigma$, consider the partition of machines into three subsets: ${\cal M}^+(\sigma)$ (overloaded), ${\cal M}^-(\sigma)$ (underloaded), and ${\cal M}^{0}(\sigma)$ (all the remaining machines). A machine $i\in {\cal M}^{+}(\sigma)$ is overloaded if $load_\sigma(i)\geq  p_{max}+\Delta+1$. A machine $i\in {\cal M}^{-}(\sigma)$ is underloaded if $load_\sigma(i)\leq \Delta$.
The remaining machines, which are neither overloaded nor underloaded, form the set ${\cal M}^{0}(\sigma)={\cal M}\smallsetminus\left({\cal M}^{-}(\sigma)\bigcup {\cal M}^{+}(\sigma)\right)$.

The procedure ${\cal UBF} (\sigma, \Delta)$ starts with some initial feasible assignment of jobs to machines and iteratively improves the makespan until it obtains an assignment with makespan of $p_{max}+\Delta$, or declares that an assignment of makespan $\Delta$ does not exist.
In each iteration, the algorithm finds an augmenting path in $G_\sigma$, from an overloaded machine to an underloaded machine, and pushes jobs along this path, by performing a series of job reassignments between machines on that path. This results in balancing the load over the machines, i.e., reducing the load of the source that is an overloaded machine, and increasing the load of the destination that is an underloaded machine, while preserving the load of all other machines.
Figure \ref{augmenting} gives a pictorial example of this operation.

$\cal UBF$ terminates when there is no path from an overloaded machine to an underloaded machine in $G_\sigma$, and this occurs after $O(mS)$ steps, where $S=\sum_{j\in {\cal J}}|\left\{ i:p_{ij}<\infty\right\} |$. Let $\tau$ be the resulting assignment after ${\cal UBF} (\sigma, \Delta)$ terminates. Then, it is shown in \cite{GL+04} that if ${\cal M}^+(\tau)\neq \emptyset$, then $T_{opt}>\Delta$. 

The procedure $\cal UBF$ combined with a binary search over the possible range of values for $\Delta$, is used to identify the smallest $\Delta$ such that a call to ${\cal UBF} (\sigma, \Delta)$ returns an assignment $\tau$ with ${\cal M}^+(\tau)= \emptyset$.
This yields the approximation ratio of $2-\frac{1}{p_{max}}$.

The running time of the approximation algorithm is factored by a value that is logarithmic in the size of the range in which we search for $\Delta$, e.g., $[0,\sum_{j\in {\cal J}}p_j]$.
Thus, the  algorithm of \cite{GL+04} computes an assignment having makespan within a factor of $2-\frac{1}{p_{max}}$ from the optimal in time $O(mSlogP)$, where $P=\sum_{j\in {\cal J}}p_{j}$.

\begin{figure}[H]
\begin{center}
{\epsfig{file=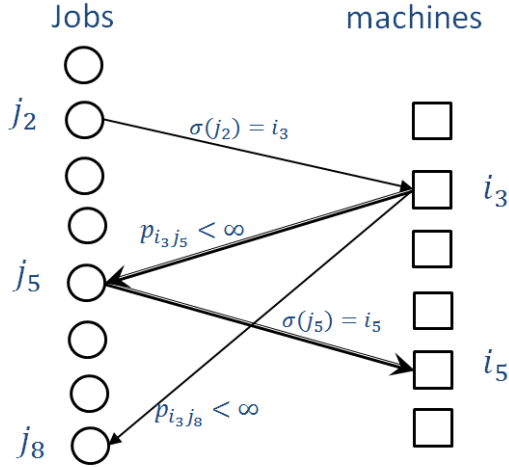,width=4.5in}}
\caption{\small The bipartite graph $G_\sigma$. By changing the orientation of the path $i_{3}\rightarrow j_{5}\rightarrow i_{5}$, we remove job $j_{5}$ from $i_{5}$ and schedule it on machine $i_{3}$. }
\label{augmenting}
\end{center}
\end{figure}

\subsection{Approximation Algorithm}
\label{sec:Approximation algorithm}

Let $\cal I$ be an instance of the restricted assignment problem.
The feasibility parameter of $\cal I$ is exactly $\varphi = \varphi({\cal I})=\frac{min_{j}\left|\left\{ i:p_{ij}<\infty\right\} \right|}{m}$.
Let $\sigma$ be an initial feasible assignment for $\cal I$, and consider the bipartite graph $G_\sigma$.

Note that any feasible assignment $\sigma$ for an instance of the restricted assignment problem has an average machine load 
\begin{equation}
\label{avg_machin_load}
\quad L(\sigma)=\frac{1}{m}\displaystyle{\sum_{j\in {\cal J}}p_{j}}
\end{equation}
Thus, $L_{opt}=\frac{1}{m}\sum_{j\in {\cal J}}p_{j}$ for any instance $\cal I$.

\comment{
We re-define underloaded, overloaded and intermediate machines.
\begin{definition}
\label{def:machinesdef}
$  $
  
${\cal M}^{-}\left(\sigma\right)=\left\{ i:load_\sigma(i)\leq \frac{L_{opt}}{\varphi} \right\}$

${\cal M}^{0}\left(\sigma\right)=\left\{ i:\frac{L_{opt}}{\varphi} <load_\sigma(i)\leq p_{max}+\frac{L_{opt}}{\varphi}\right\}$

${\cal M}^{+}\left(\sigma\right)=\left\{ i:load_\sigma(i)> p_{max}+\frac{L_{opt}}{\varphi}\right\} $
\end{definition}
}

Our algorithm proceeds as follows.
\begin{algorithm}[H]
\caption{${\cal A}_{RES}({\cal I})$}
\label{Alg_restricted}
\begin{enumerate}
\item Fix $\Delta=\lfloor \frac{L_{opt}}{\varphi({\cal I})}\rfloor$.
\item Apply ${\cal UBF}(\sigma,\Delta)$ and return the resulting assignment.
\end{enumerate}
\end{algorithm}

\comment{Throughout the execution of $\cal UBF$, augmenting paths from machines
in  ${\cal M}^{+}$ to ${\cal M}^{-}$ are found iteratively. Along each of these paths, the algorithm reassigns jobs between machines. Applying the algorithm results in reducing the makespan and balancing the loads.
The algorithm continues as long as there exists a path from ${\cal M}^{+}$ to ${\cal M}^{-}$.}

\begin{theorem}
\label{thm:approx_ratio_restricted}
For any instance $\cal I$ with feasibility parameter $\varphi({\cal I}) $ \comment{$\geq \frac {\sum_{j\in {\cal J}}p_j}{m\cdot p_{max}}$}, Algorithm \ref{Alg_restricted} yields a schedule of makespan at most $p_{max}+\frac{L_{opt}}{\varphi ({\cal I})}=p_{max}+\frac {\sum_{j\in {\cal J}}p_j}{m\cdot \varphi ({\cal I})}$, in time $O(m^2n)$.
\end{theorem}

\begin{corollary}
Let ${\cal I}$, be an instances with feasibility parameter $\varphi({\cal I})=\frac {d}{m}$, for some $d=1,...,m$. Then Algorithm \ref{Alg_restricted} yields a schedule of makespan at most $p_{max}+\frac {\sum_{j\in {\cal J}}p_j}{d}$. Therefore, for instances with sufficiently large $d$ (which is equivalent to a large feasibility parameter), namely $d$ such that $d\cdot p_{max}>\sum_{j\in {\cal J}}p_j$, we get a schedule with a makespan better than twice the optimal makespan. Let $\frac{3}{2}<r<2$, then for instances with feasibility parameter $\frac {d}{m}$ such that $d\cdot (r-1)p_{max}>\sum_{j\in {\cal J}}p_j$, Algorithm \ref{Alg_restricted} is a $r$-approximation algorithm.
\end{corollary}

The correctness of Theorem \ref{thm:approx_ratio_restricted} follows from the next lemma.

\begin{lemma}
\label{lemma:UBF}
Let $\cal I$ be an instance with feasibility parameter $\varphi({\cal I})$. Let $\sigma$ be an initial feasible assignment for $\cal I$. Then ${\cal UBF}(\sigma,\lfloor \frac{L_{opt}}{\varphi({\cal I})}\rfloor)$ terminates with ${\cal M}^{+}=\emptyset$.
\end{lemma}

\begin{proof}

Let $\tau$ be the assignment when ${\cal UBF}(\sigma,\lfloor \frac{L_{opt}}{\varphi(\cal I)}\rfloor )$ terminates. At this point,there is no path from a machine in ${\cal M}^{+}(\tau)$ to a machine in ${\cal M}^{-}(\tau)$ in the graph $G_\tau$.
Assume that ${\cal M}^{+}(\tau)\neq\emptyset$. Then there exists a machine $i'$ with $load_\tau(i')\geq  p_{max}+\lfloor \frac{L_{opt}}{\varphi} \rfloor$. Denote by ${\cal M}_{i'}$ the set of machines $i$ such that $v_i$ is reachable from $v_{i'}$, then ${\cal M}_{i'}=\{ i\in M:\mbox{there is a directed path in \ensuremath{G_\tau} from \ensuremath{v_{i'}} to \ensuremath{v_{i}} }\}$.

Obviously, there exists a job $j'$, such that $\tau(j')=i'$.
Thus, there is an edge $\left(v_{i'},u_{j'}\right)$ in $G_\tau$. Since $\varphi$ is the feasibility parameter of $\cal I$, there exists at least $\varphi m$ machines on which $u_{j'}$ is feasible, i.e., there exists an edge from $u_{j'}$ to each one of these machines. By appending each of these edges to $\left(v_{i'},u_{j'}\right)$ we get a directed path from $v_{i'}$ to at least $\varphi m$ machines (including $v_{i'}$). Therefore,
we conclude that $\left|{\cal M}_{i'}\right|\geq\varphi m$.

Now, we compute a lower bound on the average machine load for $\tau$, by summing the loads of all the machines $i\in {\cal M}$. We have that $i'\in {\cal M}^{+}(\tau)$, thus $load_\tau(i')> p_{max}+\frac{L_{opt}}{\varphi}$. Also, there is no path from $v_{i'}$ to machines in ${\cal M}^{-}(\tau)$, and therefore
${\cal M}_{i'}\cap {\cal M}^{-}(\tau)=\emptyset$. Thus, for all $i\in {\cal M}_{i'}$ it holds that 
$load_\tau(i)\geq \lfloor \frac{L_{opt}}{\varphi}\rfloor+1$. Hence,

\[
\begin{array}{ll}
\displaystyle{\sum_{i\in {\cal M}}load_\tau(i)} & \geq  load_\tau(i')+\displaystyle{\sum_{i\in {\cal M}_{i'},i\neq i'} load_\tau(i)}\\
 & > p_{max}+\lfloor \frac{L_{opt}}{\varphi}\rfloor+1+(|{\cal M}_{i'}|-1)\cdot (\lfloor \frac{L_{opt}}{\varphi}\rfloor+1)\\
  & = p_{max}+|{\cal M}_{i'}|(\lfloor \frac{L_{opt}}{\varphi}\rfloor+1)\\
  & \geq p_{max}+|{\cal M}_{i'}|((\frac{L_{opt}}{\varphi}-1)+1)\\
  & \geq p_{max}+m\varphi (\frac{L_{opt}}{\varphi})\\
  & = p_{max}+mL_{opt}\\
  & > mL_{opt} \\
\end{array}
\]

We have shown that the sum of loads of the assignment $\tau$ is greater than $mL_{opt}$. Hence, the average load for $\tau$ is greater than $L_{opt}$. A contradiction. By \ref{avg_machin_load}, the average load of any schedule, and in particular $\tau$, cannot exceed $L_{opt}$.
\end{proof}

\noindent
{\bf Proof of Theorem \ref{thm:approx_ratio_restricted}.}
Let $\cal I$ be an instance with feasibility parameter $\varphi$\comment{$\geq \frac {\sum_{j\in {\cal J}}p_j}{m\cdot p_{max}}$}.
Let $\sigma$ be some initial feasible assignment. 
By Lemma \ref{lemma:UBF}, when ${\cal UBF}(\sigma,\lfloor \frac{L_{opt}}{\varphi}\rfloor)$ terminates, ${\cal M}^{+}=\emptyset$.
Therefore, the maximum load of the resulting assignment is at most
 $p_{max}+\lfloor \frac{L_{opt}}{\varphi}\rfloor=p_{max}+\lfloor \frac{\sum_{j\in {\cal J}}p_j}{\varphi}\rfloor$. 
The running time of the algorithm equals to the running time of one call to $\cal UBF$, which is $O(mS)$, where $S=\sum_{j\in {\cal J}}|\left\{ i:p_{ij}<\infty\right\} |$. Since $S\leq m\cdot n$ for every instance $\cal I$, the algorithm terminates after $O(m^2n)$ steps.
\hfill  \qed

Note that our algorithm has better running time than the algorithm of \cite{GL+04}, since we use a single call to the procedure $\cal UBF$, in contrast to the $(2-\frac{1}{p_{max}})$-approximation algorithm of  \cite{GL+04}, which uses binary search to find the best value for $\Delta$, resulting in an overall running time of $O(mSlogP)$, where $P=\sum_{j\in {\cal J}}p_{j}$ is the sum of processing times of all jobs.

\chapter{Fixed-Parameter Algorithms for Scheduling on Unrelated Machines}
\section{Preliminaries}
Some problems can be solved exactly, or approximately, by algorithms that are exponential only in the size of a fixed parameter while polynomial in the size of the input. Such an (approximation) algorithm is called a {\em fixed-parameter tractable (FPT)} (approximation) algorithm, because the problem can be (approximately) solved efficiently for small values of the fixed parameter.

Problems in which some parameter $k$ is fixed are called parameterized problems. A parameterized problem that allows for such an FPT algorithm is said to be a {\em fixed-parameter tractable} problem and belongs to the class FPT.

We give some definitions formalizing this concept.
\begin{definition}
A problem is said to be FPT if it can be solved by an algorithm $\cal A $ that runs in time $f(k)\cdot poly(|{\cal I}|)$, for every instance $\cal I$ with parameter $k$, and where $f$ is a function independent of $|{\cal I}|$. The algorithm $\cal A$ is called an FPT algorithm.
\end{definition}

We can similarly define FPT-approximation algorithms.

\begin{definition}
A problem is said to have an $r$-FPT approximation algorithm, if there exists an $r$-approximation algorithm $\cal A$ to the problem, that runs in time $f(k)\cdot poly(|{\cal I}|)$, for every instance $\cal I$ with parameter $k$, and where $f$ is a function independent of $|{\cal I}|$.
\end{definition}

\begin{definition}
A family $\{{\cal A}_\epsilon:\epsilon>0\}$, where ${\cal A}_\epsilon$ is a $(1+\epsilon)$-FPT approximation algorithm for all $\epsilon>0$, is called a parametrized approximation scheme. 
\end{definition}

\newpage
\section{Parametrized Approximation Scheme for Scheduling on Unrelated Machines}
Consider the problem of minimizing the makespan on unrelated machines, i.e., scheduling a set $\cal J$ of $n$ jobs, $j=1,..,n$, on a set $\cal M$ of $m$
unrelated parallel machines, $i\in \cal M$, where each job $j$ has a
processing time of $p_{ij}$ on machine $i$ and the objective is to find a schedule with minimum makespan.

Our parameter $k$ of a scheduling instance is the number of machine-job
pairs $(i,j)\in {\cal M}\times {\cal J}$ such that $p_{ij}>\epsilon\cdot T$ for some $\epsilon\in(0,1]$ and a feasible value $T$. 
We will show that by rounding a solution to the MILP formulation of
the problem where the $k$ variables $x_{ij}$ such that $p_{ij}>\epsilon\cdot T$
are integral, we can get an assignment with makespan at most $(1+\epsilon)$ 
the optimal makespan.

Given a positive integer $T$, let $x_{ij}$ be an indicator to the assignment of job $j$ on
machine $i$. Consider the following mixed integer linear program.

\begin{center}
\begin{minipage}{0.95\textwidth}
\begin{eqnarray*}
MILP(\epsilon,T): 
& \displaystyle{\sum_{i=1}^{m}x_{ij}=1}, \hbox{~~~~~~~~~~ for \ensuremath{j=1,...,n} ~~~~~~~~ }  \\
& \displaystyle{\sum_{j=1}^{n}p_{ij}x_{ij}\leq T}, \hbox{~~~~~ for \ensuremath{i=1,...,m} ~~~~~~~~ } \\
& ~~~~~x_{ij}\geq0, \hbox{~~~~~~~~~~~~~~~~ for  \ensuremath{i,j} such that \ensuremath{p_{ij}\leq \epsilon \cdot T  } } \\
& ~~~~~x_{ij}\in\{0,1\}
 , \hbox{~~~~~~~~~~~ for  \ensuremath{i,j} such that \ensuremath{ p_{ij}> \epsilon \cdot T  } } \\
& ~x_{ij}=0, ~if ~ p_{ij}>T, \hbox{ for \ensuremath{i=1,...m\,,j=1,...,n}} \\
&&
\end{eqnarray*}
\end{minipage}
\end{center}

One can see that integer solutions to the above MILP are in one to one correspondence with
assignments $\sigma:{\cal J}\rightarrow{\cal M}$ of makespan at most T, and that any feasible solution to $MILP(\epsilon, T)$ has the property that the variables $x_{ij}$, such that $p_{ij}>\epsilon T$, are integral. 
\\

\begin{theorem}
Let $\cal I$ be a scheduling instance, let $\epsilon \in [0,1]$ and let $T$ be a positive
integer such that $MILP(T,\epsilon)$ is feasible. Then an assignment of makespan at
most $(1+\epsilon)T$ can be found in time $2^k\cdot poly(|{\cal I}|)$,
where $k=|\{(i,j):p_{ij}>\epsilon\cdot T\}|$.
\end{theorem}
\begin{proof}
Let $x_{ij}$, $i=1,...,m$, $j=1,...,n$ be a solution to $MILP(\epsilon,T)$.

Then for every pair $(i,j)\in {\cal M}\times {\cal J}$ such that $p_{ij}>\epsilon\cdot T$,
the corresponding $x_{ij}$ is either 0 or 1. Let $l_{i}$ be the
number of jobs $j$ for which $p_{ij}>\epsilon\cdot T$ on machine
$i$ and such that the corresponding variable $x_{ij}$ equals to $1$. Recall the rounding technique as in Theorem \ref{initial assignment thm}. Then the first $l_{i}$
slots of machine $i$ are full and therefore the capacity left in
the slots for the other jobs, with $p_{ij}\leq\epsilon\cdot T$ is
$T-\sum_{j:p_{ij}>\epsilon T}p_{ij}$. From Theorem \ref{initial assignment thm}, we get that the jobs that are assigned by the rounding to the remaining slots $l_{i}+1,...,k_{i}$ (where $k_{i}=\left\lceil \sum_{j\in {\cal J}}x_{ij}\right\rceil $) contribute at most $p_{max}^{l_{i}+1}+T-\sum_{j:p_{ij}>\epsilon T}p_{ij}$ to the load of machine $i$.
Therefore the load of machine $i$ is at most $\sum_{j:p_{ij}>\epsilon T}p_{ij}+\left(p_{max}^{l_{i}+1}+T-\sum_{j:p_{ij}>\epsilon T}p_{ij}\right)=p_{max}^{l_{i}+1}+T$.
Now, $p_{max}^{l_{i}+1}\leq\epsilon\cdot T$, and therefore the total
load of machine $i$ is at most $(1+\epsilon)T$.

The algorithm runs in time $2^k\cdot poly(|{\cal I}|)$, since the bottle-neck of the algorithm is obtaining a feasible solution to the MILP. This is done by brute-force search of at most $2^k$ possible binary values for all the variables $x_{ij}$ corresponding to $(i,j)\in S_\epsilon$, by fixing them and finding a solution for the resulting LP which can be done in polynomial-time in $|{\cal I}|$ \cite{K84}.

\end{proof}

\section{An FPT Algorithm for Graph-Balancing}
In this section we consider the special case of the restricted assignment problem, where each job can be assigned to at most two machines, with the same processing time on either machine. For this special case, Ebenlendr et al.\cite{EKS08} presented a 1.75-approximation algorithm for the minimum makespan problem.

An instance of the scheduling problem can be modeled as an undirected, multi-graph, with $m$ nodes (a node for each machine) and $n$ edges, such that every job $j$ is associated with an edge of weight $p_j$ between both machine nodes on which it can be processed, or a loop of weight $p_j$ on the only machine node on which it can be processed. Minimizing the makespan is then equivalent to the problem of {\em graph-balancing}, i.e., of orienting each edge, such that the maximum weighted in-degree over all nodes is minimized. 
We exploit this graph representation of the problem to develop an FPT algorithm for this case.

\begin{definition}
The maximum degree $r$ of an undirected graph $G=(V,E)$ is the maximum number of neighbors of any vertex, i.e., $r=max_{v\in V}|N(v)|$, where $N(v)=\{e\in E:vu,\mbox {for some } u\in V\}$.
\end{definition}

We give below an FPT algorithm for graph balancing, where the parameters are the width of the tree decomposition of the graph, and the maximum degree of the graph.

Note that we cannot hope for obtaining an FPT algorithm with the fixed parameter being only the maximum degree of the graph, as from the hardness proof for general instance \cite{EKS08}, if follows that he problem is hard to approximate within a factor less than $\frac{3}{2}$ even on bounded degree graphs, i.e., when the maximum degree is some constant.

Intuitively, a tree decomposition represents the vertices of a given graph $G$ as subtrees of a tree, in such a way that vertices in the given graph are adjacent only when the corresponding subtrees intersect.
\begin{definition}
\label{tree_decomposition_def}
Given a graph $G = (V, E)$, a tree decomposition is a pair $\left\langle X, T\right\rangle $, where $X =\{X_1,..., X_t\}$ is a family of subsets of $V$ (also called bags), and $T$ is a tree whose nodes are the subsets $X_i$, satisfying the following properties:
\begin{enumerate}
\item The union of all sets $X_i$ equals $V$. That is, each graph vertex is associated with at least one tree node.
\item For every edge $(v, w)$ in the graph, there is a subset $X_i$ that contains both $v$ and $w$. That is, vertices are adjacent in the graph only when the corresponding subtrees have a node in common.
\item If $X_i$ and $X_j$ both contain a vertex $v$, then all nodes $X_k$ of the tree in the (unique) path between $X_i$ and $X_j$ contain $v$ as well. That is, the nodes associated with vertex $v$ form a connected subset of $T$. It can be stated equivalently that if $X_i$, $X_j$ and $X_k$ are nodes, and $X_k$ is on the path from $X_i$ to $X_j$, then $X_i \cap X_j \subseteq X_k$.
\end{enumerate}
\end{definition}

The width of a tree decomposition is the size of its largest set $X_i$ minus one. The treewidth $tw(G)$ of a graph $G$ is the minimum width among all possible tree decompositions of $G$.

It is observed in \cite{B88,N06}, that many algorithmic problems that are NP-complete for arbitrary graphs can be solved efficiently by dynamic programming for graphs of bounded treewidth, using the tree decompositions of these graphs. We show that the problem of graph balancing can be solved efficiently by dynamic programming for graphs of bounded treewidth and bounded vertex degree (the maximum number of neighbors of a vertex).

\begin{theorem}
Let $G$  be an undirected edge weighted multi-graph, with given tree decomposition ${\left\langle \{X=\{X_1,...,X_t\},T\right\rangle }$ of width $w$. Then the graph balancing problem parameterized by the graph treewidth and the maximum degree, $r$, is solvable in time $O(2^{2w\cdot r}\cdot wr\cdot |X|$.
\end{theorem}

\begin{proof} 
We show how the problem can be solved using dynamic programing on the tree decomposition of $G$. \comment{We hold for every vertex $i$, a table of $2^{m_i}$ rows, and $n_i+1$ columns, where $n_i$ is the number of nodes in the bag $X_i$ and. }
The idea is to examine for each bag $X_i\in X$ all the possibilities of feasible assignments to the machines represented by the vertices in the bag $X_i$, and the jobs represented by the edges of the subgraph $G[X_{i}]$, induced by the vertices in bag $X_i$. This information is stored in a table $B_i$ corresponding to each bag $X_i$.
The tables will be updated in a post-order manner, starting at the leaves of the tree decomposition and ending at the root.
During this update process, it is guaranteed that local solutions for each subgraph corresponding to a bag of the tree decomposition are combined into a global optimal solution for the overall graph $G$.

The algorithmic details are as follows.

\paragraph{Step 0:} Set an initial orientation on the edges of the graph $G$. For each bag $X_i=\{v_{1}^{i},...,v_{n_{i}}^{i}\}$, $|X_i|=n_i$, let $E[X_{i}]=\{e_{1}^{i},...,e_{m_{i}}^{i}\}$, $|E_{G[X_{i}]}|=m_i$ . Compute the following table of $2^{m_i}$ rows, and $m_i+n_i+1$ columns: \\
\begin{center}
$B_{i}=\begin{array}{cccc||ccccc}
e_{1}^{i} & e_{2}^{i} & \cdots & e_{m_{i}}^{i} & l_i(v_{1}^{i}) & l_i(v_{2}^{i}) & \cdots & l_i(v_{n_{i}}^{i}) & T_{i}()\\
\hline
0 & 0 & \cdots & 0\\
0 & 0 & \cdots & 1\\
0 & 0 & \cdots & 1\\
\vdots & \vdots & \vdots & \vdots\\
1 & 1 & \cdots & 0\\
1 & 1 & \cdots & 1
\end{array}
$
\end{center}
 
The table consists of $2^{n_ir}$ rows and $n_ir+1$ columns. Each row represents an assignment to the sub-problem induced by the subgraph $G[X_{i}]$. Each row is a 0-1 sequence of length $n_ir$ that determines which of the edges in $G[X_{i}]$ is directed oppositely than its direction in the given initial orientation (1 if it is the opposite orientation, and 0 otherwise).  Formally, we can describe an assignment by a mapping 
\begin{equation*}
\quad A_{i}:E[X_{i}]=\{e_{i_{1}},...,e_{i_{m_{i}}}\}\rightarrow\{0,1\}.
\end{equation*}
Given the mapping $A_i$, let $In(v_j)$ denote the set of incoming edges for $v_j$, for $v_j \in X_i$.
The last column, $T_i()$, is the makespan of the assignment $A_i$.

\paragraph{Step 1:} Table initialization.

For every table $B_i$ and assignment $A_i:E[X_{i}]\rightarrow\{0,1\}$ set
\begin{equation*}
\quad l_{i}(v_j^i)(j)=\sum_{e\in In(v_j)}c(e)
\end{equation*} 

for every $v_{i_{1}},....,v_{i_{n_{i}}}\in X_i$, where $c(e)$ is the cost of edge $e$, which corresponds to the processing time of the job associated with the edge $e$.

\paragraph{Step 2:} Dynamic programming.

We now go through the tree decomposition of $G$, from the leaves to the root, and compare the corresponding tables against each other.
Let $i$ be the parent node of $j$. We show how the table for $X_i$ can be updated by the table for $X_j$.
Assume that $X_i=\{u_{1},..,u_{s},v_{1}^{i},...,v_{n_{i}-s}^{i}\}$ and $X_j=\{u_{1},..,u_{s},v_{1}^{j},...,v_{n_{j}-s}^{j}\}$, and that $E[X_i]=\{e_{1},..,e_{t},e_{1}^{i},...,e_{m_{i}-t}^{i}\}$ and $E[X_j]=\{e_{1},..,e_{t},e_{1}^{j},...,e_{m_{j}-t}^{j}\}$.

For each assignment $A:\{e_1,...,e_t\}\rightarrow\{0,1\}$, and each extension $A_i:E[X_i]\rightarrow\{0,1\}$ of $A$, we consider an assignment $A_j$, which is an extension of $A$, that minimizes the new makespan, i.e., for each assignment $A_j$, which is an extension of $A$, we calculate 
\begin{equation*}
\quad l_i^j(u_k)=l_j(u_k)+l_i(u_k)-\sum_{e\in E[X_i]\cap E[X_j]\wedge e\in In(u_k)} c(e),
\end{equation*} 
for $k=1,...,s$. Then, we calculate the makespan 
\begin{equation*}
\quad T_i^j=max\left[T_{j}(A_{j}),max_{1\leq k\leq s}l_{i}^{j}(u_{k}),max_{1\leq k\leq n_{i}-s}l_{i}(v_{k}^{i})\right]
\end{equation*} 
We update the entry for $A_i$ with $l_i^j$ and $T_i^{\hat j}$ for ${\hat j} = argmin (T_i^j)$.

\comment{
For each possible assignment of the jobs in the intersection of $i$ and $j$, and each assignment $A_j$ that is an extension of that assignment, calculate the makespan by adding the contribution of the jobs that appear in $G[X_i]$ but not in $G[X_i]$ and calculating the new makespan}

The values of $l_i(v)$, $v\in X_i$, and $T_i()$ grows by the minimal value for the makespan of the assignment problem induced by all the vertices contained in the subtree rooted at node $i$. If $i$ has several children $j_i,...,j_l$, then table $B_i $ is updated successively against all tables $B_{j_1},...,B_{j_l}$ in the same way. All this is repeated until the root node is finally updated.

\paragraph{Step 3:} Construction of a minimum makespan assignment.

The length of a minimum makespan assignment is derived from the minimal entry of the last column, $T_i()$, of the root node table. The assignment of the corresponding row shows where to assign the jobs represented by the edges in the subgraph induced by the vertices of the root bag. By recording in Step 2 how the respective minimum of each bag was determined by its children, one can easily determine the assignment of all edges in the graph.

This concludes the description of the dynamic programming algorithm. It remains to show its correctness and running time.
  
\begin{enumerate}
\item The first and second conditions in Definition \ref{tree_decomposition_def}, namely $V=\cup_{X} X_i$ and $\forall e\in E \;  \exists X_i\in X : e\in E[X_i]$, guarantee that every machine and job in the instance is considered through the computation.
\item The third condition in Definition \ref{tree_decomposition_def} guarantees the consistency of the dynamic programming. If a vertex $v\in V$ occurs in two different bags $X_{i_1}$ and $X_{i_2}$, then it is guaranteed that for the computed minimum makespan assignment only one set of jobs can be scheduled on the machine associated with that vertex $v$.
\end{enumerate}

As for the running time of the algorithm, for each edge $(X_i,X_j)$ in the tree decomposition of $G$, and for each of $2^{m_i}$ assignments $A_i$, we go over all the assignments $A_j$ that agree with $A_i$ on the edges in $E[X_i]\cap E[X_j]$ (at most $2^{m_j}$), and do a computation of time $O(n_i+m_i)$. This results in complexity of $O(2^{m_i+m_j})\cdot (n_i+m_i)\cdot |X|$. Since $m_k\leq n_k\cdot r $ and $n_k\leq w$ for all $X_k\in X$, we have that the complexity of the algorithm is $O(2^{2wr}\cdot wr\cdot |X|)$.
\end{proof}

\chapter{Reoptimization Algorithms for Scheduling Problems}

\section{Preliminaries}
Let $\Pi_{ID}$ and $\Pi_{UN}$ denote the makespan minimization problems on
identical and uniform machines, respectively. In the \emph{reoptimization} model developed in \cite{STT12}, we consider instances of the scheduling
problem that can change dynamically over time. Our goal is to compute
assignments within some guaranteed approximation for the new problem
instances, derived from the previous instances. Since the transition
from one assignment to another incurs some cost (for example, the
cost of pausing the execution of a process on one machine and resuming its execution on another),
an additional goal is to have the solution for the new instance {\em close} to
the original one (under a certain distance measure).

Let ${\cal I}_{0}=({\cal M}_{0},{\cal J}_{0}$) be an instance of jobs and machines. Let $m_0=|{\cal M}_{0}|$ and $n_0 =|{\cal J}_{0}|$ and let $\sigma_{0}:{\cal J}_{0}\rightarrow{\cal M}_{0}$
be some initial assignment for ${\cal I}_{0}$. We denote by ${\cal I}=({\cal M},{\cal J})$ a new instance derived from ${\cal I}_{0}$ by an admissible operation, e.g.,
addition or removal of jobs and/or machines. For any job $j\in{\cal J}$
and a feasible assignment $\sigma:{\cal J}\rightarrow{\cal M}$, we are
given the transition cost of $j$ when moving from the initial assignment
$\sigma_{0}$ to $\sigma$. We denote this transition cost by $c_{\sigma_{0}}(j,\sigma)$.
The goal is to find an optimal assignment for $\cal I$, for which the total
transition cost, given by $\sum_{j\in{\cal J}}c_{\sigma_{0}}(j,\sigma)$, is minimized.

Recall that, given an optimization problem $\Pi$, we denote by $R(\Pi)$ the reoptimization version of $\Pi$.

\begin{definition}
An algorithm ${\cal A}$ yields an $(r,\rho)$-reapproximation for $R(\Pi_{ID})$ ($R(\Pi_{UN}$)), for
$r, \rho \geq1$, if for any instance $\cal I$ for $\Pi_{ID}$ ($\Pi_{UN}$), ${\cal A}$ outputs an assignment of makespan at most $\rho$ times the minimal makespan for $\cal I$, and of total
transition cost at most $r$ times the minimal transition cost to an optimal assignment for $\cal I$.
\end{definition}

We consider below the case where transition costs can take values in $\{0, 1 \}$. 
In particular, job $j \in {\cJ}$ incurs a unit transition cost either if $(i)$ $j \in {\cJz}$ and is moved to a different machine in the schedule for $I$, or $(ii)$ $j \in {\cJ} \setminus {\cJz}$, i.e., $j$ is assigned to a machine for the first time
in the schedule for $\cal I$. Otherwise, the transition cost for job $j$ is equal to $0$.
Formally, $c_{\sigma_{0}}(j,\sigma)=1$ if $j\notin{\cal J}_{0}$,
or if $j\in{\cal J}_{0}$ and $\sigma_{0}(j)\neq\sigma(j)$; otherwise, $c_{\sigma_{0}}(j,\sigma)=0$.

\begin{definition}
A polynomial time reapproximation scheme (PTRS) for $R\left(\Pi\right)$ is an algorithm that, given the inputs ${\cal I}_0$ and $ \cal I$ for $R\left(\Pi\right)$ and parameters $\eps_{1},\eps_{2}\geq 0$,
yields a $\left(1+\eps_{1},1+\eps_{2}\right)$-reapproximation
for $R\left(\Pi\right)$, in time polynomial in $|{\cal I}_0|$ and $|\cal I|$.
\end{definition}

\comment{Our main result is a $(1,1+\eps)$-reapproximation algorithms
for the  reoptimization  version of the problem of scheduling on identical machines, and on uniform
machines, both with $\{0,1\}$ transition costs.
}

\section{A $(1,1+\epsilon)$-Reapproximation Algorithm for Makespan Minimization on Identical Machines}

We present below a reapproximation algorithm, ${\cal A}_{ID}$, for the problem of minimizing the makespan on identical machines.
The algorithm uses a {\em relaxed} packing of items in bins, where the items correspond to jobs, and the bins represent the set of machines.

\begin{definition}
Given a set of bins, each of capacity $K > 0$, and a set of items packed in the bins, we say that the
packing is
\emph{$\eps$- relaxed}, for some $\eps >0$, if
 the total size of items assigned to each bin is at most $(1+\eps)K$.
\end{definition}

\subsection{Algorithm}
Our algorithm for reoptimizaing makespan minimization on identical machines accepts as input
the instances ${\cal I}_0$ and $\cal I$, the initial assignment of jobs to the machines, $\sigma_0$, and an error parameter $\epsilon >0$. The algorithm proceeds as follows.

We apply a $(1+\epsilon)$-approximation algorithm \cite{H96,AA+98} on the new instance to obtain a solution of makespan $T\leq (1+\epsilon)C^{*}_{max}$. Then, we split our instance into large and small jobs, round down the large jobs processing times (to have a polynomial-size collection of feasible configurations of large jobs on the machines), such that the load of each machine does not exceed $T$. Then, we iterate on this collection in order to find the configuration that minimizes the transition cost from the original solution. We prove that after we inflate the rounded jobs to their original processing times, and greedily assign all the small jobs within the configuration, the resulting makespan is at most $(1+\epsilon)C^{*}_{max}$.  
\\

We give below a detailed description of our algorithm, ${\cal A}_{ID}$.  Let $C_{max}^*({\cal I})$ denote the minimum makespan for an instance $\cal I$. For simplicity of the presentation, for the case where $m < m_0$, we assume w.l.o.g. that the omitted machines are $m+1, m+2, \ldots , m_0$.

\begin{algorithm}[H]
\caption{${\cal A}_{ID}({\cal I}_0,{\cal I},\sigma_0)$}
\begin{enumerate}
\item[1.] \label{step:Ams_1}
Let $\epsilon_{0}=\frac{\epsilon}{4}$. Use a PTAS for makespan minimization on identical machines to find $T \leq (1+\eps_0) C_{max}^*( {\cal I})$.

\item[2.] \label{step:Ams_2} Define $\alpha_{j}=\frac{p_j}{T}$ for all $j\in{\cal J}$, and represent
each machine as a bin of unit capacity. Consider the jobs as items whose
sizes are $\alpha_j\in\left(0,1\right]$.

\item[3.] \label{step:Ams_3} An item $j \in {\cJ}$ is {\em small} if it has a size at most $\epsilon_0$; otherwise, item $j$ is {\em large}.

\item[4.] \label{step:Ams_4} Round down the sizes of the large items to the nearest multiple
of $\epsilon_0^{2}$. Denote the rounded sizes  $\bar{\alpha_j}$, for every large item $j$.

\item[5.] \label{step:Ams_5} For any feasible assignment of rounded large items in the $m$
bins, given by the configuration ${\cal C=}\{C^1, \ldots ,C^m \}$, do:

\begin{enumerate}
\item[(i)] Let $\ell=\max \{ m_0,m \} $. Construct a complete bipartite graph $G=(U,V,E)$, in which
$V = \{ 1, \ldots , \ell \}$, and $U = \{C^1,C^2, \ldots,C^\ell \}$.
Each vertex $i \in V$ corresponds to the initial configuration of machine i, given by
$C_0^i=\{j \in{\cJz}:\sigma_0(j)=i\}$, for $1 \leq i \leq m_0$; if
$m_0 <m$ , set $C_0^i=\emptyset$ for all $m_0+1 \leq i \leq \ell$.
If $m_0 \geq m$, set $C^i=\emptyset$ for all $m+1 \leq i \leq \ell$.
Define a cost on the edges $(i,C^k)$, for all $1 \leq i,k \leq \ell$, as follows.

\begin{enumerate}
\item[(a)] Add the cost of large items that appear in $C^k$ but not in the
initial configuration $C_0^i$.
\item[(b)] Add to $C^k$ all the small items that appear in $C_0^i$ but not in $C^k$ and then omit the largest small items until the total size of $C^k$ does not exceed $1$. Add the cost of the omitted items.
For an empty configuration $C^k$, the cost of the edge $(i,C_k)$, for
all $1 \leq i \leq \ell$ is equal to $0$.
\end{enumerate}
\item[(ii)] Find a minimum cost perfect matching in the bipartite graph.
\item[(iii)] Add to the solution the omitted small items using First-Fit.
\end{enumerate}
\item[6.] \label{step:Ams_6} Choose the solution of minimum cost, and return the corresponding schedule of the jobs on the machines.
\end{enumerate}
\end{algorithm}

\subsection{Analysis}

\begin{theorem}
\label{thm: Ams_PTRS}
For any $\epsilon>0$, Algorithm ${\cal A}_{ID}$
yields in polynomial time a $(1,1+\epsilon)$-reapproximation for $R\left(\Pi_{ID}\right)$,
\end{theorem}
We prove the theorem using the next lemmas.

\begin{lemma}
\label{sum of pieces lemma}
Let ${\cal I}=({\cal M},{\cal J})$ be an instance of
$\Pi_{ID}$, for which the minimum makespan is $C_{max}^{*}$, and let $T\geq C_{max}^{*}$.
Let $\alpha_{j}=\frac{p_{j}}{T}$, for all $j\in{\cal J}$; then, $\sum_{j\in{\cal J}}\alpha_{j}\leq m$.
\end{lemma}

\begin{proof}
We note that the minimum makespan satisfies $C_{max}^{*}\geq\frac{\sum_{j\in{\cal J}}p_{j}}{m}$.
Since $T\geq C_{max}^{*}$, we have that $\alpha_{j}=\frac{p_{j}}{T}\leq\frac{p_{j}}{C_{max}^{*}}$
for all $j\in{\cal J}$, therefore,
$$
m \geq \sum_{j\in{\cal J}} \frac{p_j}{T}=\sum_{j\in{\cal J}} \alpha_{j}.
$$
\end{proof}

\begin{lemma}
\label{epsilon relaxed lemma}
Let $\widetilde{{\cal C}}$ be a feasible assignment of rounded large items on the $m$
machines, given by the configuration $C=\{C^1, \ldots ,C^m \}$, to which we add in each
bin the small items that were not omitted in Step 5 of ${\cal A}_{ID}$. Then $\widetilde{{\cal C}}$ can be expanded in polynomial time to an $\epsilon_{0}$-relaxed packing of all items in the input $\cal I$.
\end{lemma}

\begin{proof}
Let ${C}= \{ C^{1},C^{2},...,C^{m} \}$ be a feasible configuration of the large
rounded items, and let $S^{1},S^{2},...,S^{m}$ be the subsets of small items added in Step 5 to
the bins, to form $\widetilde{{\cal C}}$. Then $\widetilde{{\cal C}}$ yields a feasible packing,
i.e., for all $ 1 \leq i \leq m$,
$$\sum_{j\in C^{i}}\bar{\alpha}_{j}+\underset{j\in S^{i}}{\sum}\alpha_{j}\leq 1.
$$
Now, since $\bar{\alpha}_{j}\geq\eps_{0}$ for all $j\in C^{i}$, the number of large items in bin $i$
is bounded by  $\left\lfloor 1/\eps_{0}\right\rfloor $.
Also, since $\alpha_{j}-\bar{\alpha}_{j}\leq\eps_{0}^{2}$,
for all $j\in L$, we have that
$$
\underset{j\in C^{i}\cup S^{i}}{\sum}\alpha_{j}
\leq \sum_{j\in C^i}(\bar{\alpha}_{j}+\eps_{0}^{2})+\underset{j\in S^{i}}{\sum}\alpha_{j}
\leq\underset{j\in C^{i}}{\sum}\bar{\alpha}_{j}+
\left\lfloor \frac{1}{\eps_{0}}\right\rfloor \eps_{0}^{2}+\underset{j\in S^{i}}{\sum}\alpha_{j}\leq1+\eps_{0}.
$$
Hence, the packing of $\widetilde{{\cal C}}$ is $\eps_{0}$-relaxed. Now, we
show that the packing remains $\eps_{0}$-relaxed after
we add the small items using First-Fit. Let $\a_{1},\a_{2},...\a_{t}$ denote
the sizes of the items packed in $\widetilde{{\cal C}}$, and let $\a_{t+1},...,\a_{n}$ be the sizes of the small unpacked
items. By Lemma \ref{sum of pieces lemma}, we have
\begin{equation}
\label{eq:no_1}
\sum_{i=1}^{n}\a_{j}\leq m
\end{equation}
We apply First-Fit in the following {\em relaxed} manner. Consider the next item
in the unpacked list.  Starting from bin $1$, we seek the first bin in which the item can be added,
such that the overall size of the items packed in this bin is at most
$1+\eps_{0}$. Let $r_{i}$ be the total size of the items packed in bin
$i$ after adding the small items. Assume that, after we apply First-Fit, some  small
items remain unpacked (i.e., none of the bins can accommodate these items).
Let $i_{\ell},...,i_{n}$ denote this subset of items.
Then, we have that $(1+\eps_{0})-r_{i}<\a_{j}$
for all $j=\ell,\ldots ,n$ and $i=1,\ldots,m$.
It follows, that
\begin{equation}
\label{eq:no_2}
\quad(1+\eps_{0})m-\sum_{i=1}^{m}r_{i}<m\a_{j}\qquad\forall j=\ell,\ldots,n
\end{equation}
By the definition of $r_{i}$, and since we packed all items up to $i_{\ell-1}$,
\begin{equation}
\label{eq:no_3}
\quad\sum_{i=1}^{m}r_{i}=\sum_{j=1}^{\ell-1}\a_{j}
\end{equation}
From (\ref{eq:no_1}), we get that
$$\sum\limits _{j=1}^{\ell-1}\a_{j}\leq m-\sum\limits _{j=\ell}^{n}\a_{j}.
$$
Then, from (\ref{eq:no_3}), we have
\begin{equation}
\label{eq:no_4}
\quad m-\sum\limits _{j=\ell}^{n}\a_{j}\geq\sum_{i=1}^{m}r_{i}
\end{equation}
From (\ref{eq:no_2}) and (\ref{eq:no_4}), it follows that, for all $\ell \leq j \leq n$,
$$m(1+\eps_{0})-(m-\sum\limits _{j=\ell}^{n}\a_{j})\leq m(1+\eps_{0})-\sum\limits _{i=1}^{m}r_{i}<m\a_{j},
$$
or,
\begin{equation}
\label{eq:no_5}
\quad m\eps_{0}+\sum\limits _{j=\ell}^{n}\a_{j}< m\a_{j}.
\end{equation}
From (\ref{eq:no_5}), and since $\a_{j}\leq\eps_{0}$ for $j=\ell,...,n$, we have that
$m\eps_{0}+\sum\limits _{j=\ell}^{n}\a_{j}< m\eps_{0}$; thus,
$\sum\limits _{j=\ell}^{n}\a_{j}<0$. A contradiction, since $\a_{j}\geq 0$
for all $j$.

Hence, the above relaxed implementation of First-Fit packs all of the remaining small items and yields a relaxed-packing of the
original instance.
\end{proof}

\begin{lemma}
Let $OPT$ be an optimal solution, and let ${\cal C}=\{C^{1},...,C^{m}\}$
be the configuration of large items derived from $OPT$. Then, the
cost of this optimal solution is at least the cost of the solution for $\cal C$ in Step 5 of ${\cal A}_{ID}$.
\label{ALGvsOPT lemma}
\end{lemma}

\begin{proof}
$\cal C$ is an optimal configuration, therefore $\cal C$ is also a feasible
configuration (since $\sum\limits _{j\in C^{i}}\frac{p_{j}}{T}\leq\frac{C_{max}^{*}}{T}\leq1$).
Let $ALG_{\cal C}$ be the solution for $\cal C$ the algorithm generates
in Step 5. Assume that $Cost\left(OPT\right)<Cost\left(ALG_{\cal C}\right)$. Note that the cost for the large items is the same in both solutions,
and therefore the difference between the cost is caused by packing small
items. We also note that, in the algorithm, the small items we pay
for are those that are omitted from their original bin and
transfered to a different one. Since the transition costs are 1, it
means that the number of omitted small items in $OPT$ is smaller
than in $ALG_{\cal C}$, or in other words, the number of small items
that are packed in their original bin in $OPT$ is greater than their
number in $ALG_{\cal C}$. Therefore, there must exist a bin $1\leq i\leq m$
for which this holds. Consider the small items packed in bin
$i$ in each solution. Denote by $q_{1},...,q_{s}$ the sizes of small
original items of bin $i$ that are packed in bin $i$ in both solutions,
by $p_{1}^{A},...,p_{k}^{A}$ and $p_{1}^{O},...,p_{l}^{O}$ two distinct
sets of small original items of bin $i$ that are packed in $ALG_{\cal C}$
and in $OPT$, respectively, but not in both. Since the algorithm chooses
to omit the largest small items first, the fact it chose to omit
$p_{1}^{O},...,p_{l}^{O}$ but to keep $p_{1}^{A},...,p_{k}^{A}$
means that $p_{1}^{O},...,p_{l}^{O}$ are not smaller than $p_{1}^{A},...,p_{k}^{A}$. In particular, let $p_{1}^{O}=min\left\{ p_{1}^{O},...,p_{l}^{O}\right\} $,
then we have
\begin{equation}
\label{eq:ALGvsOPT_1}
\quad p_{1}^{O}=min\left\{ p_{1}^{O},...,p_{l}^{O}\right\} \geq max\left\{ p_{1}^{A},...,p_{k}^{A}\right\}.
\end{equation}

Since $C_{max}^{*}\leq T$, we get that
\begin{equation}
\label{eq:ALGvsOPT_2}
\quad\sum_{j\in C^{i}}p_{j}+\sum_{j=1}^{s}q_{j}+\sum_{j=1}^{\ell}p_{j}^{O}\leq\frac{C_{max}^{*}}{T}\underset{}{\leq1}
\end{equation}

The algorithm chose, in particular, to omit $p_{1}^{O}$, hence 
\begin{equation}
\label{eq:ALGvsOPT_3}
\quad\sum_{j\in C_{i}}p_{j}+\sum_{j=1}^{s}q_{j}+\sum_{j=1}^{k}p_{j}^{A}\leq1+\eps_{0}
\end{equation}

but also
\begin{equation}
\label{eq:ALGvsOPT_4}
\quad\sum_{j\in C^{i}}pj+\sum_{j=1}^{s}q_{j}+\sum_{j=1}^{k}p_{j}^{A}+p_{1}^{O}>1+\eps_{0}.
\end{equation}

By the above discussion, we also have that

$\begin{array}{lcl}
\sum\limits _{j\in C^{i}}p_{j}+\sum\limits _{j=1}^{s}q_{j}+\sum_{j=1}^{k}p_{j}^{A}+p_{1}^{O}&\underset{(\ref{eq:ALGvsOPT_1})}{\leq} & \sum\limits _{j\in C^{i}}p_{j}+\sum\limits_{j=1}^{s}q_{j}+kp_{1}^{O}+p_{1}^{O}\\
& \underset{k<\ell}{\leq}&\sum\limits _{j\in C^{i}}p_{j}+\sum\limits _{j=1}^{s}q_{j}+\ell p_{1}^{O}\\
& \underset{(\ref{eq:ALGvsOPT_1})}{\leq}&\sum\limits _{j\in C^{i}}p_{j}+\sum\limits _{j=1}^{s}q_{j}+\sum_{j=1}^{\ell}p_{j}^{O}\\
& \underset{(\ref{eq:ALGvsOPT_2})}{\leq} & 1
\end{array}$

From the last inequality, we have a contradiction to ${(\ref{eq:ALGvsOPT_4})}$. Hence, we have the statement of the lemma.
\end{proof}

\begin{lemma} \cite{EK72}
Let $G=\left(V,U,E\right)$ be a bipartite graph
with $\left|V\right|=\left|U\right|=n$ and, and
let $c:E\rightarrow\mathbb{R}$ be a cost function on the edges.
A minimum cost perfect matching, i.e. a perfect matching
$M\subseteq E$ for which $\sum_{e\in M}c(e)$ is minimized, can be found in $O(n^{3})$ time.
\label{min cost perfect matching lemma}
\end{lemma}

Now, we are ready to prove the main theorem.

\paragraph{Proof of Theorem \ref{thm: Ams_PTRS}:}

Let ${\cal C}=C^{1},...,C^{m}$ be the
configuration of large items derived from an optimal solution, $OPT$. Let $ALG_{\cal C}$ be the
packing obtained for $\cal C$ in Step 5 of the algorithm. By Lemma
\ref{ALGvsOPT lemma}, we have

\begin{equation}
\label{eq:thm_Ams_PTRS_1}
\quad Cost\left(ALG_{C}\right)\leq Cost\left(OPT\right).
\end{equation}

Let $ALG$ be the solution the algorithm outputs. 

Since 
$Cost(ALG)=min\left\{ Cost(ALG_{C}):C \mbox{ is a legal configuration }\right\} $, we have that $Cost(ALG) \leq Cost(OPT)$.

Now, it is easy to see that if we transform $ALG$ to the schedule of the job on the machines (by returning to the original processing times), we get a solution of makespan
at most $(1+\eps_{0})T$. Since $T\leq(1+\eps_{0})C_{max}^{*}$,
the algorithm is a $\left(1+\eps_{0}\right)^{2}$-approximation algorithm.
We note that $(1+\eps_{0})^{2}=(1+\frac{\eps}{4})^{2}=(1+\frac{\eps}{2}+\frac{\eps^{2}}{16})\leq1+\eps$,
thus the algorithm yields a $(1+\eps)$-approximation to $C_{max}^{*}$.

We conclude that ${\cal A}_{ID}$ is a $(1,1+\eps)$-reapproximation algorithm
for the reoptimization problem.\\

Now, we prove that ${\cal A}_{ID}$ runs in polynomial time.
\begin{itemize}
\item In the first step, we run a PTAS for the new instance ${\cal I}$, therefore, this part
is polynomial in ${\cal I}$ (but exponential in $\frac{1}{\eps_{0}}$).
\item Steps 2,3 and 4 are clearly polynomial in $\left|{\cal I}\right|$.
\item For Step 5, we first show that the number of feasible configuration of large rounded items is polynomial in
${\cal I}$. Recall that the number of large rounded items in each bin can
not exceed $\left\lfloor \frac{1}{\eps_{0}}\right\rfloor $. Since the rounded sizes of large
items can be of at most $\left\lceil \frac{1}{\eps_{0}}\right\rceil $
different sizes, the number of feasible configurations is
at most $R=\binom{\left\lfloor \frac{1}{\eps_{0}}\right\rfloor +\left\lceil \frac{1}{\eps_{0}^{2}}\right\rceil }{\left\lceil \frac{1}{\eps_{0}^{2}}\right\rceil }$.
There are at most $m$ bins, implying that the number of feasible configurations is at most $\left(\begin{array}{c}
m+R\\
m
\end{array}\right)$ which is polynomial in $m$ but exponential in $\frac{1}{\eps_{0}}$. Now, for each configuration, we construct the
bipartite graph in time polynomial in $\left|{\cal I}\right|$ and by \cite{EK72} we find a perfect minimum cost matching in time $O(k^{3})$ where $k=max\left\{ m,m^{'}\right\} $, which is polynomial in $\left|{\cal I}\right|$.
\item Packing of the small items in each iteration using First-Fit is also done
in time polynomial in $\left|{\cal I}\right|$ and Step 6 is clearly polynomial in $\left|{\cal I}\right|$. 
\end{itemize}

\newpage

\section{A $(1,1+\epsilon)$-Reapproximation Algorithm for Makespan Minimization on Uniform Machines}

We present below a reapproximation algorithm, ${\cal A}_{UN}$, for the problem of minimizing the makespan on uniform machines.
The algorithm uses a {\em relaxed} packing of items in bins, where the items correspond to jobs, and the bins represent the set of machines.

\begin{definition}
Given a set of $m$ bins, with positive capacities $K_1,K_2,...,K_m$, and a set of items packed in the bins, we say that the
packing is
\emph{$\epsilon$- relaxed}, for some $\epsilon >0$, if for every $1\leq i \leq m$,
 the total size of items assigned to each bin is at most $(1+\eps)K_i$.
\end{definition}

\subsection{Algorithm}
Our algorithm for reoptimizing the makespan on uniform machines accepts as input
the instances ${\cal I}_0$ and ${\cal I}$, the initial assignment of jobs to the machines, $\sigma_0$, and an error parameter
$\epsilon >0$. 

In the previous section, it was convenient to convert the problem into a bin packing problem where all bins have equal size. We will consider the conversion to a bin packing problem also in the case of scheduling on unrelated machines, only here the bins will have variable sizes.
In the generalization of the equal-size bin case to the variable-size case, we come across a major obstacle. The size of the subintervals in which the pieces are partitioned depends on the size of the bins in which the pieces are to be packed. Moreover, the definition of large and small pieces depends on the size of the bin in which the pieces are to be packed. 
Hochbaum and Shmoys \cite{HS88} presented a PTAS for the problem of makespan minimization on uniform machines. They constructed, in polynomial time, a layered directed graph, with two nodes designated \textquotedblleft initial" and \textquotedblleft success", such that there exists a path from \textquotedblleft initial" to \textquotedblleft success" in the graph if and only if there is a schedule with makespan at most $1+\epsilon $ times the minimal makespan. From this path one can also define the configuration of \textquotedblleft medium" and \textquotedblleft large" jobs on each machine, and it is guaranteed that the \textquotedblleft small" jobs can be scheduled. We add suitable costs to the edges of the layered graph and show that a feasible solution with optimal cost can be obtained by finding the lightest path from \textquotedblleft initial" to \textquotedblleft success" in the layered graph. It is also guaranteed that if there is a path from \textquotedblleft initial" to \textquotedblleft success" in  the graph, then there is enough space in the bins to add the remaining pieces (e.g., using First-Fit), in the bins on which they are considered small. This way, as before, we can greedily pack the small pieces in their original bin (as long as we do not exceed its capacity), and pay only for packing the rest.   

\paragraph{Overview of the PTAS of \cite{HS88}} 

Assume that $s_1$ is the highest speed, and represent each machine $i$ as a bin of size $s_i$. Normalize the bin sizes by $s_1$. Consider the jobs as items whose
sizes are in $\left(0,1\right]$ (by normalizing them by some upper bound on the minimal makespan). Round down piece sizes in $\left(\eps^{k+1},\eps^{k}\right]$
to the nearest multiple of $\epsilon^{k+2}$, for some integer $k\geq 0$.

For a bin of size $s_{i}\in\left(\epsilon^{k+1},\epsilon^{k}\right]$ define:

\begin{itemize}
\item Pieces of sizes in $\left(\epsilon^{k+1},\epsilon^{k}\right]$ are {\em large} for the bin.
\item Pieces of sizes in $\left(\epsilon^{k+2},\epsilon^{k+1}\right]$ are {\em medium} for the bin.
\item Pieces of sizes less than or equal to $\epsilon^{k+2}$ are {\em small} for the bin. 
\end{itemize} 

For convenience, the interval $\left(\epsilon^{k+1},\epsilon^{k}\right]$ is referred as interval $k$.
For pieces in interval $k$, a bin is {\em large} if it is in interval $k$, {\em huge}
if it is in interval $k-1$, and {\em enormous} if it is in intervals $0,...,k-2$.

A directed layered graph, $G_L=(V_L,E_L)$, is then constructed, in which each node is labeled with a state vector describing the remaining pieces to be packed as large or medium pieces. The graph is grouped into stages, where a stage specifies the large and medium pack of bins
in one interval $\left(\epsilon^{k+1},\epsilon^{k}\right]$. Each layer within a stage corresponds to packing a bin in the corresponding interval. Both the bins within the stage and the stages are arranged
in order of decreasing bin size. 
The state vector associated with each node is of the form $(L; M; V_1, V_2, V)$, where $L$ and $M$ are vectors, each describing a distribution of pieces in the subintervals of $\left(\epsilon^{k+1},\epsilon^{k}\right]$ and $\left(\epsilon^{k+2},\epsilon^{k+1}\right]$ respectively.
There are two nodes designated \textquotedblleft initial" and \textquotedblleft success", such that \textquotedblleft initial" is connected to the initial state vectors of the first stage, and every final state vector of the final stage is connected to \textquotedblleft success".
A path from \textquotedblleft initial" to \textquotedblleft success" in $G_L$ specifies a packing of the rounded medium and large pieces for every bin.

We note that, after the large and medium pack of the bins in interval $k$, we must
allow for the packing of the remaining pieces in interval $k$ that will be packed as {\em small}.
These pieces must be packed in enormous bins for them, therefore, we need to have sufficient unused capacity in the enormous bins to at least
contain the total size of these unpacked pieces. This is represented by the value $V_1$ in
the state vector; it records the slack, or unused capacity in the partial packing of the
enormous bins with large and medium pieces. For stages corresponding to intervals
greater than $k$, we also need to have the unused capacity in the huge and large
bins, and this is the role of $V_2$ and $V$, respectively.
Since we must represent the possible values in some compact way, we consider the sizes of
the pieces that will be packed as small pieces into this as-yet-unused capacity. For $V_1$,
pieces in interval $k$ are small, and thus all pieces to be packed into this unused capacity
have rounded sizes that are multiples of $\epsilon ^{k+2}$; as a result, it will be sufficient to represent $V_l$ as an integer multiple of $\epsilon^{k+2}$. Similarly, $V_2$ and $V$ will be represented as integer multiples of $\epsilon^{k+3}$ and $\epsilon^{k+4}$, respectively.

The following lemmas, due to \cite{HS88}, will be useful in analyzing our algorithm.
\begin{lemma}
\label{graph_size lemma}
Given $\eps>0$, the layered graph $G_L$ has $O(2m\cdot n^{2/\eps^2+3}\cdot 1/\eps^6)$ nodes and\\ $O(2m(n/\eps^2)^{(2/\eps^2)+3})$ edges. The number of nodes in each layer, which is the number of state vectors corresponding to packing the bins of that layer, is $O(n^{2/\epsilon^2}(n/\epsilon^2)^3)$.
\end{lemma}

\begin{lemma}
\label{correspondence_lemma}
For any $\eps>0$, there is a one to one correspondence between paths from \textquotedblleft initial" to \textquotedblleft success" in the layered graph $G_L$ and $\epsilon$-relaxed packings.
\end{lemma}

\begin{lemma}
\label{inflating_lemma}
Given an $\eps$-relaxed packing of rounded piece sizes, for some $\epsilon>0$, restoring the piece sizes to their original size yields a $(2\eps+\eps^2)$-relaxed packing of the original pieces.  
\end{lemma}

We give below a detailed description of our algorithm, ${\cal A}_{UN}$. Let $C_{max}^*({\cal I})$ denote the minimum makespan for an instance $\cal I$.

\begin{algorithm}[H]
\caption{${\cal A}_{UN}({\cal I}_0,{\cal I},\sigma_0)$}
\begin{enumerate}
\item[1.] Let $\eps>0$. Use a PTAS for makespan minimization on uniform machines to find $T < (1+\eps)
C_{max}^*({\cal I})$.
\item[2.] Let $s_{1}$ be the largest speed, and assume that $\frac{s_1}{s_m}$ is bounded from above by some given constant $b\geq 1$. Normalize all the processing
times by $s_{1}\cdot T$, and represent each machine as a bin of capacity
$\frac{s_{i}}{s_{1}}\leq1$. Consider the jobs now as pieces with
sizes in $\left(0,1\right]$. Denote the new bin sizes by $ \frac{1}{b}\leq s_{m}\leq s_{m-1}\leq...\leq s_{1}=1$,
and denote the piece sizes by $p_{1},p_{2},...,p_{n}$.

\item[3.] \label{step:Aum_3} 
Round down the piece sizes $p_{j}\in\left(\epsilon^{k+1},\epsilon^{k}\right]$
to the nearest multiple of $\epsilon^{k+2}$.

\item[4.] \label{step:Aum_4} 
Construct the directed layered graph $G_L$, and define edge costs as follows. Consider the edge $e_\ell^k$ connecting a node of the $(\ell-1)$'th layer to a node of the $\ell$'th layer in some stage that corresponds to interval $k$. $e_\ell^k$ describes the large and medium pack of the $\ell$'th bin of stage $k$. Let ${C_0}_\ell^k$ be the set of pieces packed in that bin in the initial solution. Fix the cost on that edge to be the number of large and medium pieces that appear in ${C_0}_\ell^k$, but are not packed in the bin by $e_\ell^k$. All other edges in $G_L$ are assigned the cost zero.

\item[5.] \label{step:Aum_5} For every choice of exactly one update arc, at the end of each stage, do: 
\begin{enumerate}
\item [(i.)] Find the lightest path from \textquotedblleft initial" to \textquotedblleft success" in the graph (if one exists). Define the corresponding partial solution consisting of pieces that are packed as large or medium. 
\item [(ii.)] Add the remaining pieces greedily to the enormous bins for them: for each bin $i$, let ${S_0}_i$ be the set of remaining pieces that belong to the bin in the initial solution and that are {\em small} for this bin. Start packing the pieces in ${S_0}_i$ in bin $i$, in increasing order of piece sizes, and stop after the total size of packed pieces exceeds for the first time the bin capacity. Pack all the remaining pieces, as {\em small}, using First-Fit.
\end{enumerate}	
\item[6.] Return the solution of minimum cost.
\end{enumerate}
\end{algorithm}

\subsection{Analysis}

\begin{theorem}
\label{thm:Aum_PTRS}
For any $\eps>0$, Algorithm ${\cal A}_{UN}$
yields in polynomial time a $(1,1+\eps)$-reapproximation for $R\left(\Pi_{UN}\right)$.
\end{theorem}  

We prove the theorem using the next lemmas.
\begin{lemma}
\label{greedy_pack_um lemma}
Given the partial packing of medium and large pieces (after Step 5(i) in ${\cal A}_{UN}$), packing the remaining pieces greedily as small, in Step 5(ii), incurs the minimal cost for packing these pieces.
\end{lemma}
\begin{proof}
For the partial packing of medium and large pieces, let $R$ be the set of the remaining pieces. We show that packing the pieces in $R$ in enormous bins for them, by the greedy algorithm given in Step 5(ii), results in the minimal cost for packing $R$. The greedy algorithm first packs every bin with its original pieces that are small. It sorts them in non-decreasing order by piece size and packs them to the original bin by this order, until the bin capacity is exceeded for the first time (or all pieces are packed). This way we guarantee that we pack to every bin the maximal number of small pieces that were originally packed in it. Thus, the number of pieces, packed as small, {\em not} in their original bin, is minimized. This also implies a minimum packing cost for $R$.
\end{proof}

\begin{lemma}
Let $OPT$ be an optimal solution. Then, the
cost of $OPT$ is at least the cost of the solution obtained by ${\cal A}_{UN}$.
\label{ALGvsOPT_um lemma}
\end{lemma}

\begin{proof}
Let ${\cal C}=\{C^{1},...,C^{m}\}$
be the configuration of large and medium items derived from $OPT$.
Since $\cal C$ is optimal, in particular, it is a truly feasible configuration of large and medium pieces that also leaves enough slackness for packing the remaining pieces as {\em small}, without exceeding the capacity of the bins. Thus, by Lemma \ref{correspondence_lemma}, there is a path from \textquotedblleft initial" to \textquotedblleft success" in the layered graph, such that $\cal C$ is the large and medium pack derived from it.

This path also contains exactly one update arc after each stage. Therefore, our algorithm, in particular, considers this choice of update arcs for which it finds a lightest path in the graph (which clearly exists for this choice). by Lemma \ref{greedy_pack_um lemma}, the cost of the solution output by the algorithm is at most the cost of the lightest path plus the minimal cost of packing all the remaining pieces as {\em small}, which is clearly at most the cost of $OPT$.
\end{proof}

\begin{lemma}
For any partial solution of large and medium rounded pieces, derived from a path from  \textquotedblleft initial" to \textquotedblleft success" in $G_L$, all the remaining rounded pieces can be packed in enormous bins for them, using the greedy algorithm in Step 5(ii), such that the load on each bin of size $s_i$ is at most $s_i(1+\eps+\eps^3)$.
\label{small_pack_lemma}
\end{lemma}

\begin{proof}
As shown in \cite{HS88}, the small-pack phase, which is done after every stage, is always
successfully completed, since there is an update arc if and only if there is sufficient
total slack to accommodate all pieces to be packed as small. Using a similar argument, ${\cal A}_{UN}$ is able to pack all the remaining small pieces after packing all the large and medium pieces.
Consider the rounded pieces packed into a bin of size $s_i$,
which is in interval $k$. Focus on the small piece $j$ that, when added to the bin, exhausts
the usable slack. Then, piece $j$ is of size less than or equal to $\epsilon^{k+2}$, and before piece $j$ was added, the rounded piece sizes did not exceed $s_i+\epsilon^{k+4}$ (recall that the usable slack is rounded up to a multiple of $\eps^{k+4}$). Hence, the bin contains pieces of total (rounded) size at most $s_i+\epsilon^{k+4}+\epsilon^{k+2}$. Therefore, if $B_i$ is the set of pieces packed in bin $i$, then the sum of the rounded piece sizes in $B_i$ is at most $s_i+\epsilon^{k+4}+\epsilon ^{k+2}$ which is at most $s_i+\epsilon s_i+\eps^3s_i=s_i(1+\eps+\eps^3)$, since $s_i\geq \eps^{k+1}$.
\end{proof}

\begin{lemma}
The number of different choices for update arcs, one after each stage, is at most\\ $O((n^{2/\epsilon^2}(n/\epsilon^2 )^3)^S)$, where $S$ is the number of stages in the graph. 
\label{number_of_update_arcs_choices}
\end{lemma}
\begin{proof}
Between every two stages, there are $O(n^{2/\epsilon^2}(n/\epsilon^2)^3)$ update arcs, same as the number of state vectors in each layer (Lemma \ref{graph_size lemma}).
Therefore, the number of possibilities for choosing one arc in each stage is $O((n^{2/\epsilon^2}(n/\epsilon^2 )^3)^S)$.
\end{proof}

\begin{lemma}
\label{S_is_constant_lemma}
The number of stages $S$ in $G_L$ depends only on $\eps$ and on $b$.
\end{lemma}
\begin{proof}
The number of stages in $G_L$ is the number of intervals $k$ that contain at least one bin. Since the smallest bin size satisfies $s_m \geq \frac{1}{b}$, $S$ is at most $min\{{k\in {\mathbb N}:\eps^{k+1}<\frac{1}{b}}\}$. 
\end{proof}

Now, we are ready to prove the main theorem.

\paragraph{Proof of Theorem \ref{thm:Aum_PTRS}:} 
By Lemma \ref{ALGvsOPT_um lemma}, taking the best solution among all lightest paths (one for each choice of update arcs), we have a solution whose transition cost is at most the transition cost of an optimal solution. By Lemma \ref{small_pack_lemma}, this solution is also an $(\eps+\eps^3)$-relaxed packing of the rounded pieces, and by Lemma \ref{inflating_lemma}, after inflating the pieces to their original sizes, we get a $(2(\eps+\eps^3)+(\eps+\eps^3)^2)$-relaxed packing. Hence, by fixing a suitable initial $\eps$, e.g., $\eps=\eps_0/8$, we conclude that ${\cal A}_{UN}$ is a $(1,1+\eps_0)$-reapproximation algorithm for our reoptimization problem.\\

We note that ${\cal A}_{UN}$ runs in polynomial time. Constructing the graph $G_L=(V_L,E_L)$ can be done in time linear in the graph size, which is $V_L=O(2m\cdot n^{2/\eps^2+3}\cdot 1/\eps^6)$ and $E_L=O(2m(n/\eps^2)^{(2/\eps^2)+3})$, by Lemma \ref{graph_size lemma}. By Lemma \ref{number_of_update_arcs_choices}, the number of different choices of update arcs, is $X=O((n^{2/\epsilon^2}(n/\epsilon^2 )^3)^S)$, where $S$ is the number of stages in the graph. By Lemma \ref{S_is_constant_lemma}, $S$ depends only on $\eps$ and $b$.
Therefore, the complexity of running e.g., Dijkstra's algorithm \cite{D59} for finding a lightest path from \textquotedblleft initial" to \textquotedblleft success" and then greedily packing the small pieces, for every choice of update arcs, is $O(X\cdot (Dijkstra(G_L)+Greedy(n,m)))$.
Overall, we get that the complexity of the algorithm is $O(X\cdot (Dijkstra(G_L)+Greedy(n,m)))$, which is polynomial in the input size.

\chapter{Conclusions and Future Work}

\section{Summary}

In this work we studied the fundamental problem of makespan minimization on
parallel machines. For the unrelated machines model, we derived an improved
bound, which depends on the minimum average machine load and the feasibility
parameter of the instance. Our bound is strictly smaller than $2$, the best
known general upper bound, for a natural subclass of instances.

We further studied the power of parameterization and presented an FPT algorithm 
and a parameterized approximation scheme for instances that are known to be 
hard to approximate within factor $\frac{3}{2}$, based on classical complexity theory.

Finally, we initiated the study of the reoptimization variants of makespan
minimization on identical and uniform machines, and derived approximation
ratios that match the best known ratios in these models.

\section{Future Work}

Our study leaves open several avenues for future research.

\paragraph*{Makespan minimization on unrelated machines}
While reducing the gap between the general lower bound of $\frac{3}{2}$ and the upper
bound of $2$ remains a prominent open problem, it would be interesting to
tighten these bounds for other non-trivial subclasses of instances, either 
by a constant, or by a function of the input parameters.   

Another interesting direction is to study instances with a wider range of
feasibility parameters. In this work, we explored instances having large
feasibility parameter. On the other hand, Ebenlendr et al. \cite{EKS08} considered 
instances in which the feasibility parameter is at most $\frac{2}{m}$. Any results on
other, intermediate values of this parameter, would shed more light on its
role in obtaining better schedules.

Other natural questions are: \textquotedblleft Can we obtain better bound for fully-feasible instances, when $T_{opt}=L_{opt}$?" \textquotedblleft Can we derive a better lower bound for general instances for which $T_{opt}=L_{opt}$?"

\paragraph*{FPT algorithms for scheduling}
The problem of minimizing the makespan on unrelated machines is NP-hard even if the number of machines or number of jobs is taken as a parameter \cite{LST90}, therefore no FPT algorithm exists with only these choices of parameters. In addition, we cannot hope for obtaining an FPT algorithm for graph balancing with the fixed parameter being only the maximum degree of the graph. Indeed, by the hardness proof of \cite{EKS08}, the problem is hard to approximate within a factor less than $\frac{3}{2}$ even on bounded degree graphs, i.e., when the maximum degree is some constant.
The question whether we can find an FPT algorithm for graph balancing with the treewidth being the only parameter remains open.

\paragraph*{Reoptimization in scheduling}
We list some of the questions arising from our results. \textquotedblleft Can we obtain reapproximation algorithm with the same performance
guarantees for general instances in the uniform machines model?"
\textquotedblleft Can we extend our results to arbitrary transition costs?".
Finally, \textquotedblleft Can we improve the running times of our reapproximation schemes, by 
adjusting the known EPTASs for identical and uniform machines to the
reoptimization model?"


\newpage

\end{document}